\documentclass[11pt]{article}

\usepackage{fullpage}
\usepackage{amsmath}
\usepackage{amssymb}
\usepackage{amsthm}
\usepackage{caption}
\usepackage{hyperref, color}
\hypersetup{colorlinks=true,citecolor=red, linkcolor=red, urlcolor=red}
\usepackage[linesnumbered,lined,boxed,commentsnumbered,ruled,vlined]{algorithm2e}

\newcommand{\activeset}{\mathcal{A}}
\newcommand{\LLM}{\emph{ll-Metropolis}}
\newcommand{\DTV}[2]{d_{\mathrm{TV}}\left({#1},{#2}\right)}
\newcommand{\E}[1]{\mathbb{E}\left[{#1}\right]}
\newcommand{\dist}{\mathrm{dist}}
\newcommand{\e}{\mathrm{e}}

\newcommand{\inedge}[1]{E({#1})}
\newcommand{\bdedge}{\delta}
\newcommand{\one}[1]{\mathbf{1}\left(#1\right)}
\newcommand{\bad}{\mathcal{B}}
\newcommand{\badc}{\mathcal{C}}
\newcommand{\NeighUn}{\Gamma_{\mathsf{un}}}
\newcommand{\NeighIn}{\Gamma_{\mathsf{in}}}
\newcommand{\NeighOut}{\Gamma_{\mathsf{out}}}
\newcommand{\Red}{\mathtt{Red}}
\newcommand{\Blue}{\mathtt{Blue}}

\newcommand{\concept}[1]{\emph{{#1}}}

\newcommand{\todo}[1]{\typeout{TODO: \the\inputlineno: #1}\textbf{{\color{red}[[[ #1 ]]]}}}

\newtheorem{theorem}{Theorem}
\newtheorem{observation}[theorem]{Observation}
\newtheorem{claim}[theorem]{Claim}
\newtheorem*{claim*}{Claim}

\newtheorem{lemma}[theorem]{Lemma}

\newtheorem{corollary}[theorem]{Corollary}
\theoremstyle{definition}

\newtheorem{definition}[theorem]{Definition}

\newtheorem*{remark*}{Remark}

\title{
Distributed Symmetry Breaking in Sampling\\
{(Optimal Distributed Randomly Coloring with Fewer Colors)}\\
}
\author{
Weiming Feng~\thanks{Department of Computer Science and Technology, Nanjing University. Email: {fengwm@smail.nju.edu.cn}.}
\footnotemark[4]
\and
Thomas P. Hayes~\thanks{ Department of Computer Science, University of New Mexico. Email:{hayes@cs.unm.edu}. Partially supported by NSF CAREER award CCF-1150281.}
\and
Yitong Yin~\thanks{State Key Laboratory for Novel Software Technology, Nanjing University. Email: {yinyt@nju.edu.cn}.}
\thanks{Supported by the National Science Foundation of China under Grant No. 61672275 and No. 61722207.}
}
\date{}

\begin{document}
\maketitle

\begin{abstract}
We examine the problem of almost-uniform sampling proper $q$-colorings of a graph whose maximum degree is $\Delta$.
A famous result, discovered independently by Jerrum~\cite{jerrum1995very} and Salas and Sokal~\cite{salas1997absence}, is that, assuming $q > (2+\delta) \Delta$,
the Glauber dynamics (a.k.a.~single-site dynamics) for this problem has mixing time $O(n \log n)$, where $n$ is the number of vertices,
and thus provides a nearly linear time sampling algorithm for this problem.  A natural question is the extent to which this algorithm can
be parallelized.  Previous work~\cite{feng2017sampling} has shown that a $O(\Delta \log n)$ time parallelized algorithm is possible, and that
$\Omega(\log n)$ time is necessary.  

We give a distributed sampling algorithm, which we call the Lazy Local Metropolis Algorithm, that achieves {an optimal parallelization of this classic algorithm}. It improves its predecessor, the {Local Metropolis} algorithm of  Feng, Sun and Yin [PODC'17], {by introducing 
a step of {distributed symmetry breaking} that helps the mixing of the distributed sampling algorithm.}

For sampling almost-uniform proper $q$-colorings of graphs $G$ on $n$ vertices, we show that 
the {Lazy Local Metropolis} algorithm achieves an optimal $O(\log n)$ mixing time
if either of the following conditions is true for an arbitrary constant $\delta>0$:
\begin{itemize}
\item $q\ge(2+\delta)\Delta$, on general graphs with maximum degree $\Delta$;
\item $q \geq (\alpha^* + \delta)\Delta$, where $\alpha^* \approx 1.763$ satisfies $\alpha^* = \mathrm{e}^{1/\alpha^*}$, on graphs with sufficiently large maximum degree $\Delta\ge \Delta_0(\delta)$ and girth at least $9$.
\end{itemize}
\end{abstract}

\setcounter{page}{0} \thispagestyle{empty} \vfill
\pagebreak

\section{Introduction}
Sampling almost-uniform graph colorings is one of the most extensively studied problems in Markov chain Monte Carlo (MCMC) sampling. Let $G=(V,E)$ be a graph and $q$ a positive integer. A proper $q$-coloring $\sigma\in[q]^V$ of $G$  assigns each vertex a color from $[q]=\{1,2,\ldots,q\}$ such that no adjacent vertices receive the same color.
A classic sequential algorithm for sampling almost-uniform proper $q$-colorings is the Markov chain known as the \emph{heat bath Glauber dynamics} (a.k.a.~single-site dynamics) on proper $q$-colorings.  
For this Markov chain $(X_t)_{t\ge 0}$, each $X_t\in[q]^V$ is a $q$-coloring, and in a transition $X_t\to X_{t+1}$, a vertex $v\in V$ is chosen uniformly at random and its color $X_t(v)$ is updated to a color chosen uniformly at random from the available colors in $[q]$ that are not currently assigned by $X_t$ to $v$'s neighbors.

A famous result, discovered independently by Jerrum~\cite{jerrum1995very} and Salas and Sokal~\cite{salas1997absence}, is that, assuming $q > (2+\delta) \Delta$, where $\Delta$ is the maximum degree and $\delta>0$ is an arbitrary constant,
the Glauber dynamics defined above has mixing time $O(n \log n)$, where $n$ is the number of vertices, and thus provides a nearly linear time sequential algorithm for sampling almost-uniform proper $q$-colorings.  
For graphs with large maximum degree and large girth, 
Dyer and Frieze~\cite{dyer2001fewer} developed an approach, known as the $\emph{burn-in mehtod}$, to obtain $O(n\log n)$ mixing time of the Glauber dynamics with an improved condition $q\ge (\alpha^*+\delta)\Delta$, where $\alpha^* \approx 1.763$ satisfies $\alpha^* = \mathrm{e}^{1/\alpha^*}$.
Subsequently, the condition for the rapid mixing was improved in a series of works~\cite{vigoda2000improved, dyer2001fewer, hayes2003non, hayes2003randomly,molloy2004glauber,  hayes2006coupling, frieze2006randomly, dyer2013randomly}. See a survey of Frieze and Vigoda~\cite{frieze2007survey} for more detail.

In distributed computing, the problem of \emph{constructing} a proper $q$-coloring by local distributed graph algorithms has been extensively studied~\cite{linial1987distributive, goldberg1988parallel, awerbuch1989network, linial1992locality, szegedy1993locality, panconesi1996complexity,johansson1999simple, kuhn2006complexity, schneider2010new,barenboim2011deterministic,chung2014distributed,barenboim2016deterministic,barenboim2016locality,fraigniaud2016local}, and is a main application for {distributed symmetry breaking}~\cite{barenboim2016locality}.
These distributed algorithms assume the \emph{synchronous message-passing} model of communications. The graph $G=(V,E)$ represents a communication network.  Communications are synchronized and take place in rounds. In each round, each vertex receives messages from all neighbors, then performs the local computation, and finally sends messages to all neighbors. The {time complexity} is given by the number of rounds. {Ideally, the sizes of messages are bounded in polylogarimic of $|V|$ and the local computations are tractable.}

On the other hand, the problem of \emph{sampling} an \emph{almost-uniform} proper $q$-coloring by local distributed algorithms received much less studies. 
A natural question is the extent to which the sequential sampling algorithms can be parallelized. 

\vspace{6pt}
\noindent{\emph{Continuous-time Glauber dynamics:}}
Perhaps the most natural process to talk about as a starting point is the continuous-time Glauber dynamics.  Each vertex gets an \emph{i.i.d.}~Poisson clock with expected delay $1$; the vertex updates its color every time the clock rings.  The relationship between continuous-time and discrete-time Markov chains is well understood, and 
there are very close connections between their mixing times; see, for instance, \cite{levin2009markov}[Theorem 20.3], and \cite{hayes2007general} [Corollary 2.2]. 
In our setting, we get that the mixing time for the continuous-time dynamics is very close to being a factor $n$ speedup of the discrete Glauber dynamics.

How fast can we simulate this chain in a distributed setting?  Offhand, it looks potentially very tricky, since every now and then, there will be long chains of consecutive updates done in very short time intervals, each of which affects the next one.  However, a simple disagreement percolation argument shows that, with high probability, the continuous-time Glauber dynamics can be simulated for time $t$ in a distributed setting, with one processor for each vertex, in the time needed for $O(t \Delta / n + \log(n))$ single-vertex updates, essentially an $n/\Delta$ factor of parallel speedup.  Assuming we are in a setting, such as $q > (2 + \delta) \Delta$, in which the discrete-time Glauber dynamics has mixing time $O(n \log n)$, this implies a local distributed algorithm with running time $O(\Delta \log n)$.

\vspace{6pt}
\noindent\emph{Chromatic scheduler and systematic scans:}
A natural way to parallelize single-site dynamics is to use a chromatic scheduler to parallelize the updates, so that updates in the same round will not affect each other. 
The idea was implemented in~\cite{gonzalez2011parallel} and also by the \emph{LubyGlauber} algorithm in a previous work~\cite {feng2017sampling}. The latter achieves a $O(\Delta\log n)$ mixing time under the condition $q \ge (2+\delta) \Delta$, which is essentially due to the rapidly mixing of \emph{systematic scans}~\cite{dyer2006systematic, dyer2006dobrushin}, in which vertices are updated sequentially according to an arbitrarily fixed order. 

A fundamental issue of this type of approaches is: as a price for not allowing adjacent updates in the same round, a factor of chromatic number (or the maximum degree $\Delta$ for local distributed algorithms) is inevitably introduced to the time complexity. 

\vspace{6pt}
\noindent\emph{Local Metropolis filters:}
In a previous work~\cite{feng2017sampling}, a new parallel Markov chain, called the \emph{Local Metropolis} algorithm is introduced. It falls into the propose-and-filter paradigm of the Metropolis-Hastings algorithm. In each step, every vertex independently proposes a random color and applies a local filtration rule to accept or reject the proposals. Assuming a stronger condition on the number of colors $q\ge(2+\sqrt{2}+\delta)\Delta$, this new Markov chain achieves a $O(\log n)$ mixing time, beating the barrier of factor-$\Theta(\Delta)$ slowdown in previous approaches and achieving an  ideal factor-$\Theta(n)$ speedup of the $O(n\log n)$ mixing time of Glauber dynamics.
It was also proved in~\cite{feng2017sampling} that this $O(\log n)$ time complexity is optimal for sampling almost-uniform $q$-colorings by message-passing distributed algorithms as long as $q=O(\Delta)$. It seems that the drawback of this approach is its requirement of bigger number of colors.

\subsection{Main results}
We give a distributed MCMC sampling algorithm, called the \emph{Lazy Local Metropolis} algorithm, for sampling almost-uniform proper $q$-colorings.
The algorithm improves the \emph{Local Metropolis} algorithm in~\cite{feng2017sampling} by introducing a step of {symmetry breaking}, and achieves the optimal $O(\log n)$ mixing time while assuming smaller lower bounds on the number of colors $q$.

For sampling almost-uniform proper $q$-colorings of graphs with maximum degree $\Delta$, assuming $q\ge(2+\delta)\Delta$, the \emph{Lazy Local Metropolis} chain is rapidly mixing with rate $\tau(\epsilon)=O(\log(\frac{n}{\epsilon}))$. Note that Fischer and Ghaffari \cite{fischer2018simple} also obtain the same result independently and simultaneously. They prove this result by a different path coupling argument. See \cite{fischer2018simple} for more details.
\begin{theorem}
\label{main-theorem-1}
For any constant $\delta > 0$, for every graph $G$ on $n$ vertices with maximum degree $\Delta=\Delta_G$, if 
$q\ge(2+\delta)\Delta$, 
then given any $\epsilon>0$, the \emph{Lazy Local Metropolis} algorithm returns an almost uniform proper $q$-coloring of $G$ within total variation distance $\epsilon$ in $O(\log n+\log\frac{1}{\epsilon})$ rounds,
where the constant factor in $O(\cdot)$ depends only on $\delta$.
\end{theorem}

For graphs with large girth and sufficiently large maximum degree, by an advanced coupling similar to the one developed by Dyer~\emph{et al.}~for sequential dynamics~\cite{dyer2013randomly}, 
the condition on $q$ can be further relaxed. 
\begin{theorem}
\label{main-theorem-2}
For any constant $\delta > 0$, there exists a constant $\Delta_0 = \Delta_0(\delta)$, such that for every graph $G$ on $n$ vertices with maximum degree $\Delta=\Delta_G$ and girth $g=g(G)$, if
\begin{itemize}
\item $\Delta\ge\Delta_0$ and $g\ge 9$,
\item and $q \geq (\alpha^*+\delta)\Delta$, where $\alpha^* \approx 1.763$ satisfies $\alpha^* = \mathrm{e}^{1/\alpha^*}$,
\end{itemize}
then given any $\epsilon>0$, the \emph{Lazy Local Metropolis} algorithm returns an almost uniform proper $q$-coloring of $G$ within total variation distance $\epsilon$ in $O(\log n+\log\frac{1}{\epsilon})$ rounds,
where the constant factor in $O(\cdot)$ depends only on $\delta$.
\end{theorem}
The condition $q \geq (\alpha^*+\delta)\Delta$ matches the one achieved by Dyer~\emph{et al.}~in~\cite{dyer2013randomly} for the $O(n\log n)$-rapidly mixing of the Glauber dynamics on proper $q$-colorings of graphs with girth at least 5 and sufficiently large maximum degree.
The threshold $q \geq (\alpha^*+\delta)\Delta$ has also appeared  elsewhere variously, including:  the strong spatial mixing of proper $q$-colorings of triangle-free graphs~\cite{goldberg2005strong, gamarnik2013strong}, and rapid mixing of sequential Markov chains on proper $q$-colorings of graphs with large girth and sufficiently large maximum degree~\cite{dyer2003randomly, hayes2003randomly, hayes2006coupling}, neighborhood-amenable graphs~\cite{goldberg2005strong}, or Erd\H{o}s-R\'{e}nyi random graphs $G(n,\Delta/n)$~\cite{efthymiou2018sampling}. 

{
The \emph{Lazy Local Metropolis} algorithm in above two theorems is communication- and computation-efficient: each message consists of at most $O(\log n)$ bits and all local computations are fairly cheap.
In a concurrent work~\cite{feng2018local}, through network decomposition~\cite{ghaffari2016complexity,panconesi1996complexity}, a $O(\log^3n)$-round algorithm is given for sampling proper $q$-colorings of triangle-free graphs with maximum degree $\Delta$ assuming $q \geq (\alpha^*+\delta)\Delta$, however, with messages of unbounded sizes and unbounded local computations.
}

Theorem~\ref{main-theorem-2} is proved by establishing a so-called \emph{local uniformity property} for the Markov chain of the \emph{Lazy Local Metropolis} algorithm. Similar properties have been analyzed by Hayes~\cite{hayes2013local}  for Glauber dynamics. This is perhaps the first time this property is proved on a chain other than Glauber dynamics, not to mention a chain as a distributed algorithm.

Due to a lower bound proved in~\cite{feng2017sampling}, approximately sampling within total variation distance $\epsilon>0$ from a joint distribution with exponential decay of correlations (which is the case for uniform proper $q$-colorings as long as $q=O(\Delta)$) requires $\Omega(\log n+\log\frac{1}{\epsilon})$ rounds of communications.
Therefore, the time complexity $O(\log n+\log\frac{1}{\epsilon})$ in Theorem~\ref{main-theorem-1} and~\ref{main-theorem-2} is optimal.

\paragraph{Organization of the paper.}
Preliminaries are given in Section~\ref{sec:Preliminaries}. The \emph{Lazy Local Metropolis} algorithm is given in Section~\ref{sec:LLM}. Theorem~\ref{main-theorem-1} is proved in Section~\ref{section-coupling-on-general-graphs}. The local uniformity property is proved in Section~\ref{sec:Local-Uniformity}, with which Theorem~\ref{main-theorem-2}  is proved in Section~\ref{sec:Coupling-with-Local-Uniformity}.

\section{Preliminaries}
\label{sec:Preliminaries}

\subsection{Graph colorings}
Let $G=(V, E)$ be an undirected graph. 
For any vertex $v\in V$, we use $\Gamma(v)=\{u\mid \{u, v\}\in E\}$ to denote the set of neighbors of $v$, and $\Gamma^+(v)=\Gamma(v)\cup\{v\}$ the {inclusive neighborhood} of $v$. Let $\deg(v)=|\Gamma(v)|$ denote the degree of $v$, and $\Delta=\Delta_G=\max_{v\in V}\deg(v)$ the maximum degree of $G$. For vertices $u, v \in V$,  let $\dist(u, v)=\dist_G(u,v)$ denote the distance between $u$ and $v$ in $G$, which equals the length of the shortest path between $u$ and $v$ in graph $G$. For any integer $r\ge 0$ and vertex $v \in V$, the \concept{$r$-ball} and \concept{$r$-sphere} centered at $v$ are defined as $B_r(v) \triangleq \{u \in V \mid \dist(u, v) \leq r\}$ and $S_r(v) \triangleq \{u \in V \mid \dist(u, v) = r \}$, respectively.

Let $q$ be a positive integer. A \concept{$q$-coloring}, or just \concept{coloring}, is a vector $X\in[q]^V$. A coloring $X\in[q]^V$ is \concept{proper} if for all edges $\{u,w\}\in E$, $X(u)\neq X(v)$. For any coloring $X\in[q]^V$ and subset $S\subseteq V$, we denote by $X(S)$ the set of colors used by $X$ on subset $S$, i.e.~$X(S)\triangleq\{X(v)\mid v\in S\}$. For any two colorings $X, Y\in[q]^V$, we denote by $X \oplus Y$ the set of vertices on which $X,Y$ disagree:
\begin{align*}
X \oplus Y \triangleq \{v \in V \mid X(v) \neq Y(v)\}.
\end{align*}
The \concept{Hamming distance} between two colorings $X, Y$ is $\vert X \oplus Y \vert$.

Let $\Omega = [q]^V$ be the set of all colorings of graph $G$. 
A \concept{uniform distribution} over proper colorings of $G$ is a distribution $\mu$ over $\Omega$ such that for any coloring $X\in[q]^V$, $\mu(X)>0$ if and only if $X$ is proper; and $\mu(X)=\mu(Y)$ for any two proper colorings $X,Y$.
 
\subsection{Mixing rate and coupling}
Let $\mu$ and $\nu$ be two distributions over $\Omega$, the \emph{total variation distance} between $\mu$ and $\nu$ is defined as
\begin{align*}
\DTV{\mu}{\nu} = \frac{1}{2}\sum_{\sigma \in \Omega}\vert \mu(\sigma) - \nu(\sigma)\vert = \max_{A \subseteq \Omega}\vert \mu(A) - \nu(A) \vert.	
\end{align*}
Let $(X_t)_{t\ge 0}$ denote a Markov chain on a finite state space $\Omega$. Assume that the chain is \emph{irreducible} and \emph{aperiodic}, and is \concept{reversible} with respect to the \concept{stationary distribution} $\pi$. Then by the Markov chain Convergence Theorem~\cite{levin2009markov}, the chain $(X_t)_{t\ge 0}$ converges to the stationary distribution $\pi$. For the formal definitions of these concepts, we refer to the textbook~\cite{levin2009markov}.

Let $\pi_{\sigma}^t$ denote the distribution of $X_t$ when $X_0 = \sigma$. The mixing rate $\tau(\cdot)$ is defined as
\begin{align*}
\forall \epsilon >0:\quad \tau(\epsilon) \triangleq \max_{\sigma \in \Omega} \min\left\{t \mid \DTV{\pi_\sigma^t}{\pi} \leq \epsilon\right\}.	
\end{align*}

Let $(X_t)_{t\ge0}, (Y_t)_{t\ge0}$ be two Markov chains with the same transition rule. A \emph{coupling} of the Markov chains is a joint process $(X_t, Y_t)_{t\ge0}$ satisfying that $(X_t)$ and $(Y_t)$ individually follow the same transition rule as the original chain and $X_{t+1}=Y_{t+1}$ if $X_{t}=Y_{t}$. 
For any coupling $(X_t, Y_t)_{t\ge0}$ of the Markov chains, the total variation distance between $\pi_{\sigma}^t$ and $\pi$ is bounded as 
\begin{align*}
\max_{\sigma \in \Omega}\DTV{\pi_{\sigma}^t}{\pi} \leq \max_{X_0, Y_0 \in \Omega} \Pr[X_t \neq Y_t].
\end{align*}
The path coupling is a powerful engineering tool for constructing couplings.
\begin{lemma}[Bubley and Dyer~\cite{bubley1997path}]
\label{lemma-path-coupling}
Given a pre-metric, which is a weighted connected undirected graph on state space $\Omega$ such that all edge weights are at least 1 and every edge is a shortest path. Let $\Phi(X, Y)$ be the length of shortest path between states $X$ and $Y$ in pre-metric. Suppose that there is a coupling $(X, Y) \rightarrow (X', Y')$ of the Markov chain defined only for adjacent states $X, Y$ in pre-metric, which satisfies that
\begin{align*}
\E{\Phi(X', Y') \mid X, Y} \leq (1 - \delta)\Phi(X, Y),	
\end{align*}
for some $0 < \delta < 1$. Then the mixing rate of the Markov chain is bounded by
\begin{align*}
\tau(\epsilon) \leq \frac{1}{\delta}\log \left(\frac{\mathrm{diam}(\Omega)}{\epsilon}\right),
\end{align*}
where $\mathrm{diam}(\Omega)=\max_{X,Y\in\Omega}\Phi(X,Y)$ stands for the diameter of $\Omega$ in the pre-metric.
\end{lemma}

\section{The Lazy Local Metropolis Algorithm}
\label{sec:LLM}
In this section,  we give the {lazy local metropolis} algorithm \LLM{} for uniform sampling random proper graph coloring.

The algorithm is a Markov chain. Let $G=(V,E)$ be a graph, $q$ a positive integer and $0<p<1$. The \LLM{} \concept{chain with activeness $p$} on $q$-colorings of graph $G$, denoted as $(X_t)_{\ge 0}$, is defined as follows. Initially $X_0\in[q]^V$ is arbitrary (not necessarily a proper coloring). At time $t$, given the current coloring  $X_t\in[q]^V$, the $X_{t+1}$ is constructed as follows:
\begin{itemize}
\item Each vertex $v\in V$ becomes \emph{active} independently with probability 	$p$, otherwise it becomes \emph{lazy}. Let $\activeset\subseteq V$ denote the set of active vertices.
\item Each active vertex $v \in \activeset$ independently proposes a color $c(v) \in [q]$ uniformly at random.
\item 
For each active edge $\{u,v\} \in \inedge{\activeset}$, where $\inedge{\activeset} \triangleq \{\{u,v\} \in E \mid u\in\activeset \land v \in \activeset\}$ , we say that the edge $\{u,v\}$ passes its check if and only if $c(u) \neq c(v) \land c(u) \neq X_t(v) \land X_t(u) \neq c(v)$. For each boundary edge $\{u,v\} \in \bdedge\activeset$, where $\bdedge\activeset \triangleq \{\{u,v\} \in E \mid u \notin \activeset \land v \in \activeset\}$ and $v\in\activeset$ is active,  we say that the edge $\{u,v\}$ passes its check if and only if $c(v) \neq X_t(u)$.
\item For each vertex $v \in V$, if $v$ is active and all edges incident to $v$ passed their checks, then $v$ accepts its proposed color and updates its color as $X_{t+1}(v) \gets c(v)$; otherwise $X_{t+1}(v)\gets X_t(v)$.
\end{itemize}
The algorithm terminates after $T$ iterations and outputs $\boldsymbol{X} = (X_T(v))_{v \in V}$. The parameters $p$ and $T$ will be specified later. The pseudocode for the \LLM{} algorithm is given in Algorithm~\ref{LLM}.
\begin{algorithm}
\SetKwInOut{Input}{Input}
\Input{Each vertex $v\in V$ receives the set of colors $[q]$ and $0<p<1$ as input.}
each $v\in V$ initializes $X(v)$ to an arbitrary color in $[q]$\;
\For{$t = 1$ through $T$}{
	\ForEach{$v \in V$}{
		become active independently with probability $p$, otherwise become lazy\;	
	}
	\ForEach{active $v \in V$}{
		propose a color $c(v)\in[q]$ uniformly at random\;	
	}
	\ForEach{$\{u,v\}\in E$  where  both $u$ and $v$ are active}{
	pass the check if
	$c(u) \neq c(v) \land c(u) \neq X(v) \land X(u) \neq c(v)$\;	
	}
	\ForEach{$\{u,v\} \in E$  where $u$ is lazy and $v$ is active}{
	pass the check if
	$c(v) \neq X(u)$\;	
	}
	\ForEach{$v \in V$ and $v$ is active}{
		\If { all edges  incident to $v$ passed their checks}{
			$X(v) \gets c(v)$\;
		} 
	}
}
each $v\in V$ \textbf{returns}{ $X(v)$\;}
\caption{Pseudocode for the \LLM{} algorithm}\label{LLM}
\end{algorithm}

Compared to the Local Metropolis chain proposed in~\cite{feng2017sampling},  the \LLM{} chain allows each vertex to be lazy independently. It turns out this step is an operation of \concept{symmetry breaking} and is essential to the mixing of the parallel chain. We will see this in details in later sections.

Let $\mu$ denote the uniform distribution over proper colorings of graph $G=(V,E)$, and $\Delta$ the maximum degree of $G$.
The following theorem guarantees that the \LLM{} chain converges to the correct stationary distribution $\mu$.

\begin{theorem}
\label{theorem-LLM-convergence}
For any $ 0 < p < 1$, the \LLM{} chain with activeness $p$ is reversible with stationary distribution $\mu$, and converges to the stationary distribution $\mu$ as long as $q \geq \Delta + 2$.
\end{theorem}

\begin{proof}
First, when $q \geq \Delta + 2$, in each iteration, each vertex $v$ with positive probability becomes the only active vertex in its neighborhood and successfully updates its color. Once a vertex $v$ being successfully updated, its color will not conflict with its neighbors and will keep in that way. Therefore, when $q \geq \Delta + 2$, the \LLM{} chain is absorbing to proper colorings. 

Let $\Omega = [q]^V$ denote the state space and $P \in \mathbb{R}_{\geq 0}^{\vert \Omega \vert \times \vert \Omega \vert}$ the transition matrix for the \LLM{} chain. We will then verify the chain's irreducibility among proper colorings and aperiodicity.  For any two proper colorings $X, Y$, since $q \geq \Delta+ 2$, we can construct a finite sequence of proper colorings $X= Z_0 \to Z_1 \to \ldots \to Z_\ell = Y$, such that $Z_i$ and $Z_{i + 1}$ differ at a single vertex $v_i$. When the current coloring is $Z_i$, with positive probability, all vertices except $v_i$ are lazy, and $v_i$ proposes the color of $v_i$ in $Z_{i+1}(v_i)$, in which case the chain will move from coloring $Z_i$ to coloring $Z_{i+1}$. The chain is irreducible among proper colorings. On the other hand, due to the laziness, $P(X, X) > 0$ for all $X\in\Omega$, so the chain is aperiodic.

In the rest of the proof, we show that the following detailed balance equation is satisfied:
\begin{align}
\label{eq-detailed-balance-equation}
\forall X,Y\in\Omega:\quad \mu(X)P(X, Y) = \mu(Y)P(Y, X).	
\end{align}
This will prove the reversibility of the chain with respect to the stationary distribution $\mu$. Together with the absorption to the proper colorings, the irreducibility among proper colorings, and the aperiodicity proved above, the theorem follows according to the Markov chain convergence theorem.

If both $X, Y$ are improper colorings, then $\mu(X) =\mu(Y) = 0$, the equation holds trivially. If precisely one of $X, Y$ is proper, say $X$ is a  proper coloring and $Y$ is an improper coloring, then $X$ cannot move to $Y$ since at least one edge cannot pass its check, which implies $P(X, Y) = 0$. In both cases, the detailed balance equation holds.

Assume that $X, Y$ are both proper colorings. Consider a single move in the \LLM{} chain. Let $\activeset$ be the set of active vertices and $\boldsymbol{c} \in [q]^\activeset$ be the colors proposed by active vertices. Given the current coloring, the next coloring of \LLM{} chain is fully determined by the pair $(\activeset, \boldsymbol{c})$. 

Let $\Omega_{X \rightarrow Y}$ be the set of pairs $(\activeset, \boldsymbol{c})$ with which $X$ moves to $Y$. Given the current coloring $X$, the set of active vertices $\activeset$ and the colors $\boldsymbol{c}$ proposed by active vertices , we say a vertex $v$ is \emph{non-restricted} under the tuple $(X, \activeset, \boldsymbol{c} )$ if and only if $v$ is active and all edges incident to $v$ can pass their checks. Let $\mathcal{S}(X, \activeset, \boldsymbol{c})$ denote the the set of non-restricted vertices. Note that vertex $v$ accepts its proposed color if and only if $v \in \mathcal{S}(X, \activeset, \boldsymbol{c} ) $. Let $\Delta_{X, Y} = \{v  \in V\mid X(v) \neq Y(v)\}$ denote the set of vertices on which $X, Y$ disagree. Hence, each $(\activeset, \boldsymbol{c} ) \in \Omega_{X \rightarrow Y}$ satisfies:
\begin{itemize}
\item $\Delta_{X, Y} \subseteq \mathcal{S}(X, \activeset, \boldsymbol{c})$.
\item $\forall v \in \mathcal{S}(X, \activeset, \boldsymbol{c}): c(v) =Y(v)$.
\end{itemize}
Similar holds for $\Omega_{Y \rightarrow X}$, the set of pairs $(\activeset, \boldsymbol{c})$ with which $Y$ moves to $X$. Then we have
\begin{align}
\label{eq-ratio-P(X,Y)-P(Y, X)}
\frac{P(X, Y)}{P(Y, X)} &= \frac{\sum_{(\activeset, \boldsymbol{c}) \in \Omega_{X \rightarrow Y}} \Pr[\activeset]\Pr[\boldsymbol{c} \mid \activeset]}{\sum_{(\activeset', \boldsymbol{c}') \in \Omega_{Y \rightarrow X}} \Pr[\activeset']\Pr[\boldsymbol{c}' \mid \activeset']}
\end{align}
In order to verify the detailed balance equation, we construct a bijection $\phi_{X, Y}:  \Omega_{X \rightarrow Y} \rightarrow \Omega_{Y \rightarrow X}$, and for each pair $(\activeset, \boldsymbol{c}) \in \Omega_{X \rightarrow Y}$, denote $(\activeset', \boldsymbol{c}') = \phi_{X, Y}(\activeset, \boldsymbol{c})$. Then, we show that
\begin{align}
\label{eq-bijection}
\Pr[\activeset]\Pr[\boldsymbol{c} \mid \activeset]= 	\Pr[\activeset']\Pr[\boldsymbol{c}' \mid \activeset'].
\end{align}
Since $\mu(X) = \mu(Y)$, then combining~\eqref{eq-ratio-P(X,Y)-P(Y, X)} and~\eqref{eq-bijection} proves the detailed balance equation~\eqref{eq-detailed-balance-equation}.  

The bijection $\phi_{X, Y}:  \Omega_{X \rightarrow Y} \rightarrow \Omega_{Y \rightarrow X}$ is constructed as follows:
\begin{itemize}
\item $\activeset'=\activeset$.
\item $\forall v \in \activeset \cap \mathcal{S}(X, \activeset, \boldsymbol{c})$: since $(\activeset, \boldsymbol{c}) \in \Omega_{X \rightarrow Y}$ it must hold $c(v) = Y(u)$, then set $c'(v) = X(v)$.
\item $\forall v \in \activeset \setminus \mathcal{S}(X, \activeset, \boldsymbol{c})$: since $(\activeset, \boldsymbol{c}) \in \Omega_{X \rightarrow Y}$ it must hold $X(v) = Y(v)$, then set $c'(v) = c(v)$.
\end{itemize}
Note that the laziness and random proposed colors are fully independent. Since $\activeset = \activeset'$, then 
\begin{align*}
	\Pr[\activeset]\Pr[\boldsymbol{c} \mid \activeset] = (1-p)^{n-\vert \mathcal{A} \vert} \left( \frac{p}{q}\right)^{\vert \activeset \vert} =(1-p)^{n-\vert \mathcal{A}' \vert} \left(\frac{p}{q}\right)^{\vert \activeset' \vert}  = \Pr[\activeset'] \Pr[\boldsymbol{c} \mid \activeset' ],
\end{align*}
which proves equation~\eqref{eq-bijection}. We finish the proof of reversibility by showing that $\phi_{X, Y}$ is indeed a bijection from $\Omega_{X \rightarrow Y}$ to  $\Omega_{Y \rightarrow X}$.  

Consider the move from $X$ to $Y$ with pair $(\activeset, \boldsymbol{c} )$. For each edge $\{u,v\} \in \inedge{\activeset} \cup \bdedge\activeset$, define the indictor variable $\mathsf{pass}(u,v, X, \activeset, \boldsymbol{c})$ indicating whether edge $\{u,v\}$ passes its check under the tuple $(X, \activeset, \boldsymbol{c})$. Note that $X(u) \neq X(v)$ because $X$ is a proper coloring, we have
\begin{align*}
\mathsf{pass}(uv, X, \activeset, \boldsymbol{c}) &= \begin{cases}
 	\one{c(u) \neq c(v)}\one{c(u) \neq X(v)}\one{X(u) \neq c(v)}  &\text{if } uv \in \inedge\activeset\\
 	\one{c(v) \neq X(u)} &\text{if } uv \in \bdedge\activeset\text{ and }v\in\activeset
 \end{cases}\\
&= \begin{cases}
 	\prod_{x \in \{c(u), X(u)\} \atop y \in \{c(v), X(v)\}}\one{x \neq y} &\text{if } uv \in \inedge\activeset\\
 	\prod_{x \in \{c(v), X(v)\}}\one{x \neq X(u)} &\text{if } uv \in \bdedge\activeset\text{ and }v\in\activeset
 \end{cases}.
\end{align*}
Similarly, for each edge $\{u,v\} \in E(\activeset') \cup \delta \activeset'$, note that $\activeset' = \activeset$, we have
\begin{align*}
\mathsf{pass}(uv, Y, \activeset', \boldsymbol{c}') &= \begin{cases}
 	\prod_{x \in \{c'(u), Y(u)\} \atop y \in \{c'(v), Y(v)\}}\one{x \neq y} &\text{if } uv \in \inedge\activeset\\
 	\prod_{x \in \{c'(v), Y(v)\}}\one{x \neq Y(u)} &\text{if } uv \in \bdedge\activeset\text{ and }v\in\activeset
 \end{cases}.
\end{align*}
According to the definition of $\phi_{X, Y}$, it must hold that $\{c(v), X(v)\} = \{c'(v), Y(v)\}$; $\{c(u), X(u)\} = \{c'(u), Y(u)\}$ (if $u$ is active); $X(u) = Y(u)$ (if $u$ is lazy), which implies $\mathsf{pass}(uv, X, \activeset, \boldsymbol{c}) = \mathsf{pass}(uv, Y, \activeset', \boldsymbol{c}')$. Hence, it holds that $\mathcal{S}(X, \activeset, \boldsymbol{c}) = \mathcal{S}( Y, \activeset', \boldsymbol{c}')$, with which we can easily verify that $(\activeset', \boldsymbol{c}') \in \Omega_{Y \rightarrow X}$ and $\phi_{X, Y} = \phi^{-1}_{Y,X}$. This proves that $\phi_{X,Y}$ is a bijection from $\Omega_{X \rightarrow Y}$ to $\Omega_{Y \rightarrow X}$. 

\end{proof}

\section{Mixing When $q\ge(2+\delta)\Delta$ on General Graphs}
\label{section-coupling-on-general-graphs}
In this section, we analyze the mixing time for the \LLM{} chain for proper $q$-colorings on general graphs. 
We show that the chain mixes within $ O\left(\log n \right)$ rounds under the Dobrushin's condition $q \geq (2 + \delta)\Delta$, even when the maximum degree $\Delta$ is unbounded.
\begin{theorem} 
\label{theorem-LLM-worst-case-coupling}
For all $\delta > 0$, there exists $C=C(\delta)$, such that for every graph $G$ on $n$ vertices with maximum degree $\Delta$, if $q \geq (2 + \delta)\Delta$, then the mixing rate of the \mbox{\LLM{}} chain with activeness $p= \min\left\{\frac{\delta}{3},\frac{1}{2}\right\}$ on $q$-colorings of graph $G$ satisfies 
\[
\tau(\epsilon) \le C\log \frac{n}{\epsilon}.
\]	
\end{theorem}
The mixing rate is proved by a path coupling with respect to the Hamming distance. 
Compared to the coupling based on disagreement percolation for a non-lazy version of the chain in~\cite{feng2017sampling}, where the disagreement may percolate to distant vertices within one step, our new coupling is local, as within one step the disagreement can at most contaminate the adjacent vertices. 
Thanks to the symmetry breaking due to the independent laziness, this local coupling achieves a much better mixing condition than the one achieved in~\cite{feng2017sampling} with a much shorter analysis.

\paragraph{The local coupling:}
Assume $X, Y \in [q]^V$ to be two colorings (not necessarily proper) that differ only at one vertex $v_0$. Without loss of generality, we assume $X({v_0})=\Red$, 
$Y({v_0})=\Blue$.

We then construct a coupling $(X, Y) \rightarrow (X', Y')$. Given the current coloring $X$, the random coloring of the next step $X'$ is determined by the random choice of $(\activeset_X, \boldsymbol{c}_X)$ where $\activeset_X$ is the set of active vertices and $\boldsymbol{c}_X \in [q]^{\activeset_X}$ is the vector of colors proposed by active vertices. The coupling of the chain $(X, Y) \rightarrow (X', Y')$ is then specified by a coupling
of the random choices $(\activeset_X, \boldsymbol{c}_X)$ and $(\activeset_Y, \boldsymbol{c}_Y)$ of the two chains, which is described as follows:
\begin{enumerate}
\item 
First, the laziness is coupled identically.
Each vertex $v\in V$ becomes active in both chains, independently with probability $p$. 
Let $\activeset=\activeset_X=\activeset_Y$ denote the set of active vertices.
\item
Then, the random proposals $(\boldsymbol{c}_X, \boldsymbol{c}_Y)$  for the active vertices in $\activeset$ are coupled step by step as follows. Recall that $\Gamma(v_0)$ denotes the set of neighbors of $v_0$.
\begin{enumerate}
\item \label{item:consistent-couple} For every active vertex $v\not\in\Gamma(v_0)$, the random proposals $(c_X(v),c_Y(v))$ are coupled \textbf{identically} such that $c_X(v)=c_Y(v)=c(v)\in[q]$ is sampled uniformly and independently.
\item For every active vertex $v\in\Gamma(v_0)$, if at least one of the following conditions is satisfied, the random proposals $(c_X(v),c_Y(v))$ are coupled \textbf{identically}:
\begin{itemize}
\item for at least one of $v$'s neighbor $u\neq v_0$, the current color satisfies that $X(u)=Y(u)\in\{\Red,\Blue\}$;
\item  for at least one of $v$'s active neighbor $u\not\in\Gamma(v_0)$, the random proposal already sampled as in Step~\ref{item:consistent-couple} has $c_X(u)=c_Y(u)\in\{\Red,\Blue\}$.
\end{itemize}
For all other active vertices $v\in\Gamma(v_0)$, the random proposals $(c_X(v),c_Y(v))$ are coupled identically except with the roles of $\Red$ and  $\Blue$ \textbf{switched} in the two chains.
\end{enumerate}
\end{enumerate}
With the random choices $(\activeset_X, \boldsymbol{c}_X)$ and $(\activeset_Y, \boldsymbol{c}_Y)$ coupled as above, the colorings $(X',Y')$ of the next step are constructed following the rules of the \LLM{} chain described in Algorithm~\ref{LLM}.

It is easy to verify this is a valid coupling of the \LLM{} chain, as in each individual chain $X$ or $Y$, each vertex $v$ becomes active independently with probability $p$ and proposes a  random color $c(v)\in[q]$ uniformly and independently.

The following observations for the coupling can be verified by case analysis.
\begin{observation}
\label{observation-worst-case-coupling}
The followings hold for the coupling constructed above:
\begin{itemize}
\item For each vertex $u \neq v_0$ that 
$X_u = Y_u \in \{\Red, \Blue\}$, all its active neighbors $w \in \Gamma(u) \cap \activeset$ sample $(c^X_w, c^Y_w)$ consistently.
\item For each active vertex $u \in \Gamma(v_0)$, $c^X_u=c^Y_u$ if and only if there exists a vertex $w \in \Gamma(u)$, such that $X_w = Y_w \in \{\Red, \Blue\}$ or $c^X_w = c^Y_w \in \{\Red, \Blue\}$. 
\item For each vertex  $u \neq v_0$, the event $X'_u \neq Y'_u$ occurs only if $u\in\activeset$ and $\{c^X_u, c^Y_u\} \subseteq \{\Red, \Blue\}$.	
\end{itemize}
\end{observation}

\begin{proof}
The first two observations are easy to verify. 
We prove the last one. 

If $u$ is lazy, then $X'_u = X_u = Y_u =  Y'_u$ holds trivially. We then assume that vertex $u$ is active.
Supposed  $\{c^X_u, c^Y_u\} \not \subseteq  \{X_{v_0}, Y_{v_0}\}$, then regardless of which distribution $(c^X_u, c^Y_u)$ is sampled from, it must hold that $c^X_u = c^Y_u \not \in \{X_{v_0}, Y_{v_0}\}$. Supposed $c^X_u = c^Y_u \not \in \{X_{v_0}, Y_{v_0}\}$, we prove that each edge $uw \in E$ passes its check in chain $X$ if and only if $uw$ passes its check in chain $Y$. Note that $X_u = Y_u$. This implies the contradictory result $X'_u = Y'_u$.

There are two cases for vertex $u \neq v_0$:
\begin{itemize}
\item \textbf{Case: $X_u = Y_u \in \{X_{v_0}, Y_{v_0}\}$.} In this case, by the first observation, for each $w \in \Gamma(u)$, it holds that
\begin{enumerate}
\item either $\{X_w,  Y_w\} = \{X_{v_0}, Y_{v_0}\}$ or $X_w = Y_w$;
\item if $w$ is active, then $c^X_w = c^Y_w$.
\end{enumerate}
Since we assume that  $c^X_u = c^Y_u \not \in \{X_{v_0}, Y_{v_0}\}$, then edge $uw$ passes its check in chain $X$ if and only if $uw$ passes its check in chain $Y$.
\item \textbf{Case: $X_u = Y_u \not \in  \{X_{v_0}, Y_{v_0}\}$.} In this case, since permuted distribution only swaps the roles of $X_{v_0}$ and $Y_{v_0}$, then for each $w \in \Gamma(u)$. It holds that
\begin{enumerate}
\item either $\{X_w,  Y_w\} = \{X_{v_0}, Y_{v_0}\}$ or $X_w = Y_w$;
\item if $w$ is active, then either $\{c^X_w,  c^Y_w\} = \{X_{v_0}, Y_{v_0}\}$ or $c^X_w = c^Y_w$.
\end{enumerate}
Since we assume that  $c^X_u = c^Y_u \not \in \{X_{v_0}, Y_{v_0}\}$, then edge $uw$ passes its check in chain $X$ if and only if $uw$ passes its check in chain $Y$.
\end{itemize}
\end{proof}

The following lemma bounds the  discrepancy at each vertex in $(X',Y')$.
 \begin{lemma}
 \label{lemma-worst-case-coupling-x'v-neq-y'v}
 For vertex $v_0$ at which the two colorings $X, Y \in [q]^V$ differ, it holds that
 \begin{align}
 \label{eq-worst-case-coupling-v0}
 \Pr[X'({v_0}) = Y'({v_0)} \mid X, Y] \geq \frac{p(q - \Delta)}{q}\left(1 - \frac{3p}{q}\right)^\Delta.
 \end{align}
For any vertex $u \in \Gamma(v_0)$, it holds that
\begin{align}
\label{eq-worst-case-coupling-u}
\Pr[X'(u) \neq Y'(u) \mid X, Y] \leq \frac{p}{q}.	
\end{align}
For any vertex $w \in V \setminus \Gamma^+(v_0)$, it holds that
\begin{align}
\label{eq-worst-case-coupling-w}
\Pr[X'(w) \neq Y'(w) \mid X, Y] =0.	
\end{align}
\end{lemma}

\begin{proof}
The event $X'({v_0}) = Y'({v_0})$ occurs if following events occur simultaneously:
\begin{itemize}
\item Vertex $v_0$ is active, which happens with probability $p$.
\item $c_X(v_0) \not \in \{X(u) \mid u \in \Gamma(v_0)\}$ (hence $c_Y({v_0}) \not \in \{Y(u) \mid u \in \Gamma(v_0)\}$ due to $c_X(v_0) = c_Y(v_0)$ and $X(u) = Y(u)$ for all $u \in \Gamma(v_0)$). Since $v_0$ has at most $\Delta$ neighbors, this event occurs with probability at least $\frac{q - \Delta}{q}$ conditioning on the occurrence of the previous event.
\item For every vertex $u \in \Gamma(v_0)$, either $u$ is lazy in both chains or $c_X(u) \not \in \{\Red, \Blue, c_X(v_0)\}$ (hence regardless of whether $(c_X(u), c_Y(u))$ is coupled identically or with $\Red/\Blue$ switched, it must hold that $c_Y(u) \not \in \{\Red, \Blue, c_Y({v_0})\}$ by the coupling). Since each vertex becomes lazy and proposes color independently and $v_0$ has at most $\Delta$ neighbors, this event occurs with probability at least 
$\left(1 - p + p \frac{q - 3}{q}\right)^{\Delta} = \left(1 - \frac{3p}{q}\right)^{\Delta}$
conditioning on the occurrences of previous events.
\end{itemize}
Inequality~\eqref{eq-worst-case-coupling-v0} then follows by the chain rule.

For each $u \in \Gamma(v_0)$, by Observation~\ref{observation-worst-case-coupling}, the event $X'(u) \neq Y'(u)$ occurs only if $u$ is active and $\{c_X(u), c_Y(u)\} \subseteq \{\Red, \Blue\}$. Vertex $u$ becomes active with probability $p$. Assuming that $u$ is active, we prove inequality~\eqref{eq-worst-case-coupling-u} by exhausting the two cases:
\begin{itemize}
\item \textbf{Case 1:} $(c_X(u), c_Y(u))$ are coupled identically. Note that the event $X'(u) \neq Y'(u)$ occurs only if $\{c_X(u), c_Y(u)\} \subseteq \{\Red, \Blue\}$. However,  by the part two of Observation~\ref{observation-worst-case-coupling},  there must exist $w \in \Gamma(u)$ such that $X(w) = Y(w) \in\{\Red, \Blue\}$ or $c_X(w) = c_Y(w) \in \{\Red, \Blue\}$. Without loss of generality,  assume that $X(w) = Y(w) = \Red$ (other cases follow by symmetry). If $c_X(u) = c_Y(u) = \Red$, then the edge $\{u,w\}$ cannot pass its check in either chain, which implies $X'(u) = X(u) = Y(u) = Y'(u)$. Thus, the event $X'(u) \neq Y'(u)$ occurs with probability at most $\frac{1}{q}$ conditioning on $u$ being active.
\item \textbf{Case 2:} $(c_X(u), c_Y(u))$ are coupled with the roles of $\Red/\Blue$ switched. Note that the event $X'(u) \neq Y'(u)$ occurs only if $\{c_X(u), c_Y(u)\} \subseteq \{\Red, \Blue\}$. However, if $c_X(u) = \Red=X(v_0)$ and $c_Y(u) =\Blue= Y({v_0})$, then the edge $\{u,v_0\}$ cannot pass its check neither in chain $X$ nor in chain $Y$, which implies $X'(u) = X(u) = Y(u) = Y'(u)$. Thus, the event $X'(u) \neq Y'(u)$ occurs with probability at most $\frac{1}{q}$ conditioning on $u$ being active.
\end{itemize}
Combining the two cases we have the inequality~\eqref{eq-worst-case-coupling-u}.

Now we prove~\eqref{eq-worst-case-coupling-w}.
If $w$ is at distance $3$ or more from $v_0$, then for all vertices $u \in \Gamma^+(w)$, it holds that $X(u) = Y(u)$; and furthermore, for all vertices $u \in \Gamma^+(w) \cap\activeset$, it holds that $c_X(u) = c_Y(u)$, which implies $X'(w) = Y'(w)$.
If $w$ is at distance $2$ from $v_0$, then by Observation~\ref{observation-worst-case-coupling}, the event $X'(w) \neq Y'(w)$ occurs only if $w$ is active and $\{c_X(w), c_Y(w)\} \subseteq \{\Red, \Blue\}$. Note that $w$ must propose color identically in the two chains. If $c_X(w) = c_Y(w) \in \{\Red, \Blue\}$, then by the coupling all vertices $u \in \Gamma^+(w) \cap \activeset$ must propose color identically in the two chains. Note that for all vertices $u \in \Gamma^+(w)$, it holds that $X(u) = Y(u)$. Combining them together we have $X'(w) = Y'(w)$.
\end{proof}

\begin{proof}[Proof of Theorem~\ref{theorem-LLM-worst-case-coupling}]
Combining~\eqref{eq-worst-case-coupling-v0}, \eqref{eq-worst-case-coupling-u} and~\eqref{eq-worst-case-coupling-w} in Lemma~\ref{lemma-worst-case-coupling-x'v-neq-y'v} together and due to linearity of expectation, we have
\begin{align*}
\E{\vert X' \oplus Y' \vert \mid X, Y} &= \sum_{v \in V} \Pr[X'(v) \neq Y'(v) \mid X, Y]\\
& = \Pr[X'(v_0) \neq Y'(v_0) \mid X, Y] + \sum_{u \in \Gamma(v_0)}\Pr[X'(u) \neq Y'(u) \mid X, Y]\\
&\leq 1 - \frac{p(q - \Delta)}{q}\left(1 - \frac{3p}{q}\right)^\Delta + \frac{p\Delta}{q}\\
(q \geq (2 + \delta)\Delta)\qquad 
&\leq 1 - p \left(\frac{1+\delta}{2+\delta} \left(1 - \frac{3p}{(2+\delta)\Delta}\right)^\Delta - \frac{1}{2+ \delta} \right)\\
\left(\text{Assume }p \le 1/2\right)\qquad&\leq 1 - p \left(\frac{1+\delta}{2+\delta} \left(1 - \frac{3p}{2+\delta}\right) - \frac{1}{2+ \delta} \right).
\end{align*}
The last inequality is due to Bernoulli's inequality $(1+x)^r \geq 1 + rx$ for $r \geq 1$ and $x \geq -1$. For $p=\min\{\frac{\delta}{3},\frac{1}{2}\}$, it holds that
\begin{align*}
\E{\vert X' \oplus Y' \vert \mid X, Y} \le 
\begin{cases}
1 - \frac{\delta^2}{3(2+\delta)^2} & \text{ if }\delta\le \frac{3}{2},\\
1-\frac{2\delta^2-\delta}{4(2+\delta)^2}	 & \text{ if }\delta> \frac{3}{2}.
\end{cases}
\end{align*}
The Hamming distance between two colorings is at most $n$. By the path coupling lemma~\ref{lemma-path-coupling}, the mixing rate is $\tau(\epsilon) = O \left(\log n+ \log \frac{1}{\epsilon} \right)$,  where the constant in $O(\cdot)$ depends only on $\delta$.
\end{proof}

\section{Local Uniformity for Parallel Chain}
\label{sec:Local-Uniformity}

In this section, we establish the so-called \concept{local uniformity} property for the \LLM{} chain, with which we can prove Theorem~\ref{main-theorem-2} i.e. the mixing condition with few colors in graphs with large girth and large maximum degree.  

To properly state this property for colorings, we need to define the notion of available colors.
\begin{definition}
Let $G=(V,E)$ be a graph, and $X \in [q]^V$ an arbitrary coloring, not necessarily proper. For any vertex $v\in V$, the set of \concept{available colors} at $v$ under coloring $X$ is defined as
\begin{align}
A(X, v) = [q] \setminus X(\Gamma(v)),\label{eq:available-colors}
\end{align}
where $X(\Gamma(v)) = \{X_u \mid u \in \Gamma(v)\}$ is the set of colors used by $v$'s neighbors in the coloring $X$.
\end{definition}
Inequality~\eqref{eq-worst-case-coupling-v0} of the worst-case path coupling in last section can be generalized to: 
\begin{align*}
\Pr[X'(v_0) = Y'(v_0)\mid X,Y] \geq \frac{p\cdot |A(X, v_0)|}{q}\left(1 - \frac{3p}{q}\right)^\Delta,
\end{align*}
where the inequality~\eqref{eq-worst-case-coupling-v0} is actually obtained by applying this general inequality with the naive bound $|A(X,v_0)|\ge q-\Delta$ for the worst case colorings $X,Y$.

When the current coloring $X$ is produced by a Markov chain, especially after running for a while, it is conceivable that the number of available colors $|A(X,v)|$ at each vertex $v$ with high probability is much bigger than this worst case lower bound, and is closer to that in a uniform random coloring, which is $\approx q \e^{-\deg(v) / q}$. This is guaranteed by the local uniformity properties established for the respective chains.
More precisely, the local uniformity properties are a number of ``local'' properties of graph coloring which holds with high probability for a uniformly random coloring~\cite{hayes2013local}. Here in particular, what we need is the lower bound on the number of available colors. 
The following theorem states a local uniformity for the \LLM{} chain on graphs with girth at least 9 and sufficiently large maximum degree.
\begin{theorem}[\LLM{} local uniformity]
\label{theorem-local-uniformity-G}
For all $\delta>0$, $0<\zeta <\frac{1}{10}$, $0 < p < \frac{1}{2}$, there exists $\Delta_0 = \Delta_0(p, \delta,\zeta)$, $C = C(\delta, \zeta)$,  such that 
for all graphs $G=(V,E)$ with maximum degree $\Delta \geq \Delta_0$ and girth at least 9, all $q \geq (1 + \delta)\Delta$, 
the following holds.
Let $(X_t)_{t \geq 0}$ be the \LLM{} chain with activeness $p$ for $q$-colorings on graph $G$.
For any $v \in V$, 
\begin{align*}
\Pr\left[ \forall t \in \left[  t_0 , t_\infty \right]: \,\, \frac{\vert A(X_t, v) \vert }{q} \ge (1 - 10\zeta)\e^{-\deg(v) / q} \right] \ge 1- \exp(-\Delta / C),
\end{align*}
where $t_0=\frac{1}{p}\left(\frac{1 + \delta}{\delta}\right)^2\ln\frac{1}{\zeta}$ and $t_\infty=\exp(\Delta / C)$.
\end{theorem}

This is the first local uniformity result proved for a parallel chain. In fact, to the best of our knowledge, all previous local uniformity results were established for Glauber dynamics. Compared to typical local uniformity results~\cite{hayes2013local,efthymiou2016convergence}, the parallel chain acquires the local uniformity much faster: after $O(1)$ steps instead of $O(n)$ steps, and a $t_\infty=\exp(\Delta / C)$ (instead of $n\exp(\Delta / C)$) is sufficient for applying the local uniformity in proving the mixing rate. Meanwhile, we need a bigger girth ($\ge 9$) to deal with the local dependencies between adjacent vertices in the parallel chain.

The rest of this section is dedicated to the proof of this theorem.
 
\subsection{The \LLM{} chain on a modified graph $G^*$}
In order to prove the local uniformity property in Theorem~\ref{theorem-local-uniformity-G}, we construct a modified graph $G^*$ and define a \LLM{} chain on the modified graph $G^*$. We will show a local uniformity property for this process on $G^*$. Then Theorem~\ref{theorem-local-uniformity-G} can be proved by comparing the original \LLM{} chain on $G$ with this modified process on $G^*$.

Consider an undirected graph $G=(V,E)$ with girth at least 9. Fix any vertex $v \in V$. The graph $G^*$ is a \concept{mixed graph}, meaning that it has both directed and undirected edges. The mixed graph $G^*$ is obtained by replacing all the undirected edges within the ball of radius 4 centered at $v$ with directed edges towards $v$. Since the girth of $G$ is at least 9, each directed edge has a unique direction. The remaining edges in graph $G$ are preserved and kept undirected in $G^*$. 
\begin{definition}
\label{definition-graph-G*}
Let $r\ge 1$ and $G=(V, E)$ an undirected graph with girth at least $2r+1$.
Fix any vertex $v\in V$. Let $G_{\mathsf{in}}(v, r)$ denote the mixed graph $G^*=(V, E^*, F^*)$ with vertex set $V$, undirected edge set $E^*$, and directed edge set $F^*$, where
\begin{itemize}
\item $E^* = \{\{u, w\}\in E \mid \dist_G(v, u)> r \lor \dist_G(v, w)> r \lor \dist_G(v, u)=\dist_G(v, w)=r\}$,
\item $F^* = \{( u, w) \mid \{u, w\} \in E \land \dist_G(v, w) < \dist_G(v, u) \leq r \}$.	
\end{itemize}
In particular, let $G=(V, E)$ be an undirected graph with girth at least $9$. Fix an arbitrary $v\in V$. We define $G^*=G_{\mathsf{in}}(v, 4)$.
\end{definition}
For any vertex $u$ in graph $G^*$, we define 
\begin{align*}
\NeighUn(u) &\triangleq \{w \mid \{u, w\} \in E^*\}, \\
\NeighIn(u) &\triangleq \{ w \mid ( w, u ) \in F^*\},\\
\NeighOut(u) &\triangleq \{ w \mid ( u, w ) \in F^*\}.
\end{align*}
We have $\Gamma(u) = \NeighUn(u) \cup \NeighIn(u) \cup \NeighOut(u)$ for the set  of neighbors $\Gamma(u)$ of $u$ in $G^*$ (and also in $G$).

The \LLM{} chain $(X_t^*)_{t\ge 0}$ on $q$-colorings of graph $G^* = G_{\mathsf{in}}(v, 4)$ is defined as follows. Initially, $X^*_0\in[q]^V$ is arbitrary.
Given the current coloring $X^*_t \in [q]^V$, the $X^*_{t+1}$ is obtained as:
\begin{itemize}
\item Each vertex $u \in V$ becomes active independently with probability $p$, otherwise it becomes lazy. Let $\mathcal{A}^*\subseteq V$ denote the set of active vertices.
\item Each active vertex $u \in \activeset^*$ independently proposes a color $c^*(u) \in [q]$ uniformly at random.
\item For each vertex $u \in \activeset^*$, for each $w \in \Gamma(u)$, we say that \emph{the pair $(u,w)$ passes the check initiated at $u$} if and only if
\begin{align*}
\begin{cases}
c^*(u) \neq c^*(w) \land c^*(u) \neq X^*_t(w) \land X^*_t(u) \neq c^*(w)	 &\text{ if } w \in\activeset^*\text{ and } w\in\NeighUn(u)\cup \NeighIn(u),\\
c^*(u) \neq X^*_t(w)	 &\text{ if } w \not\in\activeset^*\text{ and } w\in\NeighUn(u)\cup \NeighIn(u),\\
c^*(u) \neq c^*(w)  \land X^*_t(u) \neq c^*(w)	 &\text{ if } w \in\activeset^*\text{ and }w\in\NeighOut(u),\\
\text{always pass check} &\text{ if } w \not\in\activeset^*\text{ and }w\in\NeighOut(u).
\end{cases}
\end{align*}
\item Let $\mathcal{R}^* \subseteq \activeset^*$ denote the subset of active vertices $u$ such that $\forall w \in \Gamma(u)$, the pair $(u,w)$ passed the check initiated at $u$. The coloring $X_{t+1}\in[q]^V$ at time $t+1$ is constructed as
\begin{align*}
X_{t+1}^*(u) = \begin{cases}
 c^*(u) &\text{ if } u \in \mathcal{R}^*,\\
 X_t^*(u) &\text{ if } u \not \in \mathcal{R}^*.
 \end{cases}
\end{align*}
\end{itemize}
Note that the original \LLM{} chain in Algorithm~\ref{LLM} can be seen as a special case of the above process when $\Gamma(u)=\NeighUn(u)$ and $\NeighIn(u)=\NeighOut(u)=\emptyset$ for every vertex $u\in V$.

The only differences between this new Markov chain $(X_t^*)_{t\ge 0}$ on $G^*$ and the original \LLM{} chain $(X_t)_{t\ge 0}$ on graph $G$ are the trimmed local Metropolis filters on outgoing directed edges. 
Consider a directed edge $( u, w )$ in graph $G^*$. Vertex $u$ updates its color oblivious to the current color of vertex $w$. This makes the \LLM{} chain on graph $G^*$ not reversible, and may move from proper colorings to improper ones. 
Nevertheless, this \LLM{} chain on graph $G^*$ has two nice features. 
First, the random colors assigned to $u \in \Gamma(v)$ are conditional independent, which helps establishing the local uniformity property (proved in Section~\ref{sec:LLM-local-uniformity-G*}).
Second, there is a coupling between this new process and the original \LLM{} chain on $G$ that preserves the local uniformity (proved in Section~\ref{sec:LLM-local-uniformity-comparison}).

\subsection{Local uniformity for the \LLM{} chain on $G^*$}\label{sec:LLM-local-uniformity-G*}
We prove a local uniformity property for the \LLM{} chain on the modified graph $G^*$, in terms of the lower bound on the number of available colors. For the mixed graph $G^*$, we override the definition of the set of available colors $A(X,u)$ in~\eqref{eq:available-colors} by assuming $\Gamma(u) = \NeighUn(u) \cup \NeighIn(u) \cup \NeighOut(u)$.

\begin{lemma}
\label{lemma-local-uniformity-G*}
For all $\delta, \ell, \zeta > 0,  0 < p < 1/2$, there exists $\Delta_1 = \Delta_1(\ell,p, \delta,\zeta)$, such that for all graphs $G=(V,E)$ with maximum degree $\Delta \geq \Delta_1$ and girth at least 9, all $q \geq (1 + \delta)\Delta$, the following holds.  
Fix any vertex $v\in V$ and let $G^* = G_{\mathsf{in}}(v, 4)$.

Let $(X^*_t)_{t \geq 0}$ be the \LLM{} chain with activeness $p$ for $q$-colorings on graph $G^*$.
\begin{align*}
\Pr\left[\vert A(X_\ell^*, v) \vert \ge (q -\ell)\left(\frac{1-\gamma}{\mathrm{e}}\right)^{\deg(v)/(q - \ell)} - \zeta q \right] \ge 1- \exp(-\zeta^2 q/ 2),
\end{align*}
where 
\begin{align*}
\gamma = \exp{\left( - p \left(\frac{\delta}{1+\delta}\right)^2\ell\right)}+\frac{1}{q}\left(\frac{1+\delta}{\delta}\right)^2.
\end{align*}
\end{lemma}

To prove the lemma, we need Chernoff bounds of various forms.
\begin{theorem}[Chernoff bound]
\label{thm-Chernoff-bound}
Let $X_1,X_2,\ldots,X_n \in \{0,1\}$ be mutually independent or negatively associated random variables, let $X = \sum_{i = 1}^n X_i$ and $\mu = \E{X}$. For any $\delta >0$, it holds that
\begin{align}
\label{eq-Chernoff-bound-1}
\Pr[X \geq (1+\delta)\mu] \leq \left( \frac{\e^\delta}{(1+\delta)^{(1+\delta)}} \right)^\mu \leq \left( \frac{\e}{1+\delta} \right)^{(1+\delta)\mu}.	
\end{align}
In particular, if $k \geq \e^2\mu$, then
\begin{align}
\label{eq-Chernoff-bound-2}
\Pr[X \geq k] \leq \exp(-k).
\end{align}
Let $t > 0$, it holds that
\begin{align}
\label{eq-Chernoff-bound-3}
\Pr[X \leq \mu -t] \leq \exp\left(-\frac{2t^2}{n}\right).
\end{align}
\end{theorem}

We also need the following technical  lemma due to Dyer and Frieze~\cite{dyer2001fewer}, and refined by Hayes~\cite{hayes2013local} for proving the local uniformity property for Glauber Dynamics. 
We use a slightly modified version here, which says that for a sequence of independent random colors, if there exists a subset of colors, in which no color is very likely to be sampled in any step, then with high probability, there are many missed colors. 
The proof is very similar to the ones in~\cite{dyer2001fewer} and~\cite{hayes2013local}, which we include here for completeness. 

\begin{lemma}[Dyer and Frieze]
\label{lemma-available-color-concentration}
Let $q, s$ be positive integers, and let $c_1, \ldots,c_s$ be independent (but not necessarily identically distributed) random variables taking values in $[q]$. Let $S \subseteq [q]$ with size $\vert S \vert  = m$. Suppose that there is a $\gamma<1$ such that $\Pr[c_i = j] \leq \gamma$ for every $1\le i\le s$ and $j\in S$.
Let $A = [q] \setminus \{c_1 \ldots,c_s\}$ be the set of missed colors. Then
\begin{align*}
\E{\vert A \vert} \geq m(1 - \gamma)^{s / m\gamma} \geq m \left( \frac{1-\gamma}{\e} \right)^{s/m},
\end{align*}
and for every $a > 0$,  $\Pr\left[\vert A \vert \leq \E{\vert A \vert} - a\right] \leq \mathrm{e}^{-a^2/2q}$.
\end{lemma}

\begin{proof}
For each $ 1\leq i \leq s$ and $1\leq j \leq q$, let random variable $\eta_{ij}$ indicate  the event $c_i = j$, thus
\begin{align*}
\vert A \vert = \sum_{j = 1}^q \prod_{i = 1}^s (1 - \eta_{ij} )	.
\end{align*}
By the linearity of expectation and the independence of the colors $c_i$, we have
\begin{align*}
\E{\vert A \vert} &= \sum_{j = 1}^q \prod_{i = 1}^s (1 - \E{\eta_{ij}})\\
(0 \leq \eta_{ij} \leq 1) \qquad &\geq	\sum_{j \in S}\prod_{i = 1}^s (1 - \E{\eta_{ij}})\\
(\text{AM-GM inequality})\qquad&\geq m \prod_{i = 1}^s\prod_{j \in S} (1 - \E{\eta_{ij}})^{\frac{1}{m}}.
\end{align*}
Note that for each $1\leq i \leq s$, it holds that $\sum_{j \in S} \E{\eta_{ij}} \leq \sum_{j = 1}^q  \E{\eta_{ij}} = 1$ and for all $j \in S$, $ 0 \leq \E{\eta_{ij}} \leq \gamma$.
For each $1 \leq i \leq s$, s uppose that $\sum_{j \in S}\E{\eta_{ij}} = \beta_i \leq 1$, the minimum of the $\prod_{j \in S}(1 - \E{\eta_{ij}})$ is achieved when as many as possible of $\E{\eta_{ij}}$ equal $\gamma$, hence
\begin{align*}
\prod_{j \in S}(1 - \E{\eta_{ij}}) &\geq (1-\gamma)^{\lfloor \beta_i / \gamma \rfloor}(1 - (\beta_i - \gamma\lfloor \beta_i / \gamma \rfloor ))\\
(r = \beta_i / \gamma  - \lfloor\beta_i/\gamma\rfloor < 1)\qquad&=\frac{(1-\gamma)^{\beta_i / \gamma}}{(1-\gamma)^r}(1-\gamma r)\\
(\ast) \qquad&\geq (1 - \gamma)^{\beta_i /\gamma }\\
(\beta_i \leq 1)\qquad&\geq (1-\gamma)^{1/\gamma},
\end{align*}
where $(\ast)$ is due to Bernoulli's inequality $(1+x)^t \leq 1 + xt$ when $0\leq t \leq 1$ and $x \geq -1$.
Thus
\begin{align*}
\E{\vert A \vert} &\geq m(1-\gamma)^{s/m\gamma}
\geq m\left( \frac{1-\gamma}{\mathrm{e}} \right)^{s/m}.
\end{align*}

For each $1\leq i \leq s$, since $\sum_{j=1}^q \eta_{ij} = 1$, the random 0-1 variables $\eta_{i1},\eta_{i2},\ldots,\eta_{iq}$ are negatively associated. Since the color choices are mutually independent, then all $\eta_{ij}$ are negatively associated. Because decreasing functions of disjoint subsets of a family of negatively associated variables are also negatively associated \cite{dubhashi1998balls}, the  $q$ random variables $\left\{\prod_{i = 1}^s(1 - \eta_{ij})\right\}_{(1 \leq j \leq q)}$ are negatively associated.
Then by the Chernoff bound for negatively associated variables~\eqref{eq-Chernoff-bound-3}, for every $a > 0$ it holds that $\Pr[\vert A \vert \leq \E{\vert A \vert} - a] \leq \e^{-a^2/2q}$.
\end{proof}

\begin{proof}[Proof of Lemma~\ref{lemma-local-uniformity-G*}]
Recall that $B_r(v)\subseteq V$ and $S_r(v) \subseteq V$ denote the $r$-ball and $r$-sphere centered at vertex $v$ in graph $G$, which contains the same set of vertices as the $r$-ball in $G^*$. And we use notation $\Gamma^+(v)$ denote $\Gamma(v) \cup \{v\}$.

Let $\mathcal{F}$ denote the random choices for the laziness and proposed colors of all vertices in $\left(V \setminus B_2(v)\right) \cup \{v\}$ during the time interval $[1, \ell]$ in the chain $(X_t^*)_{t\ge 0}$.  Note that given any $\mathcal{F}$, the followings hold.
\begin{itemize}
\item For all vertices in $V  \setminus B_3(v)$,  the whole procedure of \LLM{} on graph $G^*$ during the time interval $[0,\ell]$ is fully determined. Because the procedure outside the ball $B_3(v)$ requires no information in $B_3(v)$ except the laziness and the random proposed colors of vertices in $S_3(v)$, which are given by condition $\mathcal{F}$.
\item The laziness and random proposed colors of vertex $v$ are given by condition $\mathcal{F}$. 
\item The subgraph reduced by $B_3(v)$ is a tree because the girth of graph is at least 9. 
\end{itemize}
Hence, given the condition $\mathcal{F}$, for each vertex $u \in \Gamma(v)$, the random color $X^*_\ell(u)$ only depends on the random choices of laziness and proposed colors of vertices $w \in \Gamma^+(u) \setminus \{v\}$ during $[1,\ell]$. Since the laziness and proposed colors are fully independent, then given condition $\mathcal{F}$,  the neighbor colors $X_\ell^*(u)$ for $u \in \Gamma(v)$ are conditionally fully independent.
 
Next, we describe the conditional distribution of $X^*_\ell(u)$ given $\mathcal{F}$, where $u$ is a neighbor of $v$. Let $\mathcal{S}_\mathcal{F}$ be the set of colors proposed by vertex $v$ during the time interval $[1, \ell]$, which is uniquely determined by the condition $\mathcal{F}$. For each color $c \in [q] \setminus \mathcal{S}_\mathcal{F}$, we bound the probability of the event $X^*_\ell(u) = c$. We say \emph{vertex $u$ successfully updates its color at step $t$} if and only if $u$ accepts its proposed color at step $t$. The event $X^*_\ell(u) = c$ occurs only if one of following two events occurs.
\begin{itemize}
\item Vertex $u$ never successfully updates its color in time interval $[1, \ell]$ and $X^*_0(u) = c$. Then, in each step $t$, $X^*_t(u) = c$. Note that $c \not \in \mathcal{S}_\mathcal{F}$, which implies the color proposed by $v$ cannot coincide with color $c$. Thus, the event that $u$ successfully updates its color at step $t$ occurs if following three events occur simultaneously : 
\begin{enumerate}
\item vertex $u$ is not lazy at step $t$, which occurs with probability $p$;
\item vertex $u$ proposes a color $\sigma$ such that $\sigma \not \in X^*_{t-1}(\Gamma(w) \setminus \{v\} )$ and $\sigma$ does not coincide with the color proposed by $v$ if $v$ is not lazy at step $t$, which occurs with probability at least $(q-\Delta)/q$ condition on previous event;
\item each vertex $w \in \Gamma(u) \setminus \{v\}$ either becomes lazy or does not propose $X^*_{t-1}(u) = c$ or $\sigma$,  which occurs with probability at least $(1-2p/q)^\Delta$ condition on previous events;
\end{enumerate}
Thus, the probability that $u$ successfully updates its color at each step $t$ is at least 
\begin{align*}
\frac{p(q - \Delta)}{q}\left(1 - \frac{2p}{q}\right)^\Delta &\geq \frac{p \delta}{1+\delta}\left(1 - \frac{2p}{(1+\delta)\Delta}\right)^\Delta\\
(\ast)\qquad&\geq \frac{p \delta}{1+\delta}\left(1 - \frac{2p}{(1+\delta)}\right)\\
(p < 1/2)\qquad &\geq p \left(\frac{\delta}{1+\delta}\right)^2,
\end{align*}
where $(\ast)$ is because Bernoulli's inequality $(1+x)^r \geq 1 + rx$ for $r \geq 1$ and $x \geq -1$. Hence,  the probability  of the event that $X_{0 }^*(u) = c$ and $u$ never successfully updates its color in the time interval $[1,\ell]$ is at most
$\left(1  - p \left(\frac{\delta}{1+\delta}\right)^2\right)^\ell \leq \exp{\left( - p \left(\frac{\delta}{1+\delta}\right)^2\ell\right)}$.	

\item Vertex $u$ successfully updates its color in the time interval $[1, \ell]$, and at the last time when $u$ successfully updates its color, $u$ updates it into color $c$. For each $1 \leq i \leq \ell$,  let $\mathcal{U}_i$ be the event that vertex $u$ successfully updates its color into $c$ at time $i$ and $u$ never makes any successful update during $[i + 1, \ell]$. Then this event is $\bigcup_{1\leq i \leq \ell}\mathcal{U}_i$. The event $\mathcal{U}_i$ occurs only if vertex $u$ is not lazy and proposes color $c$ at step $t$ and $u$ never successfully updates its color during $[i+1,\ell]$. The event $\mathcal{U}_i$ implies  $X_t(u) = c$ for all $i+1 \leq t \leq \ell$, thus we have
\begin{align*}
\Pr\left[\mathcal{U}_i \mid \mathcal{F} \right] \leq \frac{p}{q}\left(1  - p \left(\frac{\delta}{1+\delta}\right)^2\right)^{\ell -i}.
\end{align*}
Take a union bound over all $1 \leq i \leq \ell$, we have
\begin{align*}
\Pr\left[\bigcup_{1 \leq i\leq 1} \mathcal{U}_i \mid \mathcal{F}\right] \leq \frac{p}{q}\sum_{i = 1}^{\ell}\left(1  - p \left(\frac{\delta}{1+\delta}\right)^2\right)^{\ell - i} \leq \frac{1}{q}\left(\frac{1+\delta}{\delta}\right)^2.	
\end{align*}
\end{itemize}
Combine two cases together and use the union bound, we have
\begin{align*}
\Pr[X^*_\ell (u) = c \mid \mathcal{F}] \leq \exp{\left( - p \left(\frac{\delta}{1+\delta}\right)^2\ell\right)}+\frac{1}{q}\left(\frac{1+\delta}{\delta}\right)^2.	
\end{align*}
Recall that the above probability bound holds for all color $c \in [q] \setminus \mathcal{S}_\mathcal{F}$. Note that $\left\vert  [q] \setminus \mathcal{S}_\mathcal{F} \right\vert \geq q - \ell$ because the size of $\mathcal{S}_\mathcal{F}$ is at most $\ell$. Apply  Lemma~\ref{lemma-available-color-concentration} with $\gamma = \exp{\left( - p \left(\frac{\delta}{1+\delta}\right)^2\ell\right)}+\frac{1}{q}\left(\frac{1+\delta}{\delta}\right)^2$, $a = \zeta q$, $m = q - \ell$ and $s = deg(v)$. Note that if we take $\Delta > \frac{1+\delta}{\delta^2 \left( 1-\exp\left( -p\delta^2\ell/(1+\delta)^2 \right) \right)}$, then $\gamma < 1$; if we take $\Delta > \frac{\ell}{1+\delta}$, then $m > 0$. Thus, for $\Delta$ sufficiently large, we have
\begin{align*}
\Pr\left[\vert A(X_\ell^*, v) \vert \leq (q -\ell)\left(\frac{1-\gamma}{e}\right)^{deg(v)/(q - \ell)} - \zeta q \mid \mathcal{F}\right] \leq \exp(-\zeta^2 q / 2).
\end{align*}
Finally,  by the law of total probability, summing over all conditions $\mathcal{F}$ yields
\begin{align*}
\Pr\left[\vert A(X_\ell^*, v) \vert \leq (q -\ell)\left(\frac{1-\gamma}{e}\right)^{deg(v)/(q - \ell)} - \zeta q \right] \leq \exp(-\zeta^2 q/ 2).
\end{align*}
\end{proof}

\subsection{Comparison of the \LLM{} chains}\label{sec:LLM-local-uniformity-comparison}
Next, we show that there is a coupling between the \LLM{} chains respectively on $G$ and $G^*$ that preserves the local uniformity.
\begin{lemma}
\label{lemma-compare-G-vs-G*}
For all $C, \delta, \zeta> 0$, $0< p < 1$, there exists $\Delta_2 = \Delta_2(C, p,  \delta, \zeta)$, such that for all graphs $G=(V,E)$ with maximum degree $\Delta \geq \Delta_2$ and girth at least 9, all $q \geq (1 + \delta)\Delta$, the following holds. 
Fix any vertex $v\in V$ and let $G^* = G_{\mathsf{in}}(v, 4)$.
Let $(X_t)_{t \geq 0}$ and $(X_t^*)_{t \geq 0}$ be the \LLM{} chains with activeness $p$ for $q$-colorings  on $G$ and $G^*$ respectively, where $X_0=X_0^*\in[q]^V$. 
There exists a coupling $(X_t, X^*_t)_{t\ge 0}$  of the processes $(X_t)_{t \geq 0}$ and $(X_t^*)_{t \geq 0}$ such that
\begin{align*}
\Pr\left[\forall t \leq C, \forall u\in V:\,\,  \vert (X_t \oplus X^*_t) \cap \Gamma(u) \vert \leq \zeta\Delta\right] \geq 1 - \exp(-\Delta).	
\end{align*}
\end{lemma}
\begin{proof}
Define the identical coupling $(X_t, X^*_t)$ as follows: In each step, two chains sample the same active vertex set $\activeset = \activeset^*$ and all active vertices propose the same random colors $\boldsymbol{c} = \boldsymbol{c}^*$.	

Let random variable $D_{\leq t} = \bigcup_{t' \leq t}(X_{t'} \oplus X^*_{t'})$ denote the  set of all disagreeing vertices appeared before time $t$. We prove the Lemma by showing that with probability at least $1 - \exp(-\Delta)$,  for all $u \in V$, it holds that $\vert D_{\leq C} \cap \Gamma(u) \vert \leq \zeta\Delta$.

Let $R = C + 6$. The bad events $\bad_1$ and $\bad_2$ defined as follows:
\begin{align*}
\bad_1:\quad&D_{\leq C} \not \subseteq B_{R-3}(v);\\
\bad_2:\quad&\vert D_{\leq C} \vert \geq \Delta^{13/4}.
\end{align*}
we then show that $\bad_1$ can never occur and $\bad_2$ occurs with very small probability. 

We begin by showing that disagreements cannot percolate outside the ball $B_{R-3}(v)$, i.e.: 
\begin{align}
\label{eq-bad-1}
\Pr[\bad_1] = 0.
\end{align}
Consider any vertex $u \notin B_4(v)$. Since no directed edge is incident to $u$ in graph $G^*$, then the updating rule for vertex $u$ is identical in two chains. Vertex $u$ becomes a new disagreeing vertex at step $t$ only if there exists a vertex $w \in \Gamma(u)$, such that $X_{t-1}(w) \neq X^*_{t-1}(w)$. Since $X_0 = X^*_0$, then it must hold that $D_{\leq 1} \subseteq B_4(v)$. Furthermore, for all $t \geq 1$, it must hold that $D_{\leq t} \subseteq B_{t+3}(v)$. In particular, $D_{\leq C} \subseteq B_{C+3}(v)$, which implies~\eqref{eq-bad-1}.

To bound the probability of bad event $\bad_2$, consider the random variable 
$$\mathcal{N} (D_t) = \vert(X_t \oplus X^*_t) \setminus (X_{t-1} \oplus X^*_{t-1}) \vert,$$
which gives the number of new disagreements contributed at time $t$. Any $u$ with $X_{t-1}(u) = X^*_{t-1}(u)$ but becomes a disagreement at time $t$ only if it is incident to following two types of \emph{bad edges}:
\begin{itemize}
\item \emph{Type-1 bad edge:} An undirected edge $\{u, w\} \in E^*$ or a directed edge $( w, u )  \in F^*$ such that $X_{t-1}(w) \neq X^*_{t-1}(w)$. For such bad edges, $u$ becomes a new disagreement only if $u$ is active at time $t$ and proposes 	$X_{t-1}(w)$ or $X^*_{t-1}(w)$.
\item \emph{Type-2 bad edge:} A directed edge $( u, w) \in F^*$. For such bad edges, $u$ becomes a new disagreement only if $u$ is active at time $t$ and proposes 	$X_{t-1}(w)$. In this case, the pair $(u,w)$ may pass the check initiated at vertex $u$ in chain $X^*$ but the undirected edge $\{u,w\}$ can not pass the check in chain $X$. 
\end{itemize}
Suppose vertex $u$ is incident to $k$ bad edges, then the probability that $u$ becomes a new disagreement is at most $2kp/q$. Since the maximum degree is at most $\Delta$ and $\vert F^* \vert \leq \Delta^4$, then the total number of bad edges is at most $\Delta \vert D_{\leq t-1} \vert + \Delta^4$. Hence, we have 
\begin{align*}
\E{\mathcal{N} (D_t) \mid D_{\leq t-1}} \leq \frac{2p(\Delta \vert D_{\leq t-1} \vert + \Delta^4)}{q}	\leq \frac{2p( \vert D_{\leq t-1} \vert + \Delta^3)}{1+\delta}.
\end{align*}
Furthermore, the laziness and proposed colors are mutually independent, which implies $\mathcal{N} (D_t)$ is stochastically dominated by the sum of independent random 0-1 variables. Then by the Chernoff bound~\eqref{eq-Chernoff-bound-2}, together with $\vert D_{\leq t} \vert \leq \vert D_{\leq t - 1} \vert + \mathcal{N}(D_t)$  we have
\begin{align*}
\Pr\left[ \vert D_{\leq t} \vert \geq \vert D_{\leq t-1} \vert + \frac{20p( \vert D_{\leq t-1}\vert + \Delta^3)}{1+\delta} \mid D_{\leq t - 1}\right] &\leq \Pr\left[\mathcal{N} (D_{t}) \geq \frac{20p(\vert D_{\leq t-1}\vert + \Delta^3)}{1+\delta} \mid D_{\leq t-1} \right]\\
\text{(Chernoff bound)}\qquad&\leq \exp\left( -\frac{20p(\vert D_{\leq t-1}\vert + \Delta^3)}{1+\delta} \right)\\
(\ast)\qquad&\leq \exp(-\Delta^2),
\end{align*}
where $(\ast)$ is because if we take $\Delta \geq \frac{1+\delta}{20p}$, then $\frac{20p(\vert D_{\leq t-1}\vert + \Delta^3)}{1+\delta} \geq \frac{20p\Delta^3}{1+\delta}\geq \Delta^2$.
Hence, with probability at least $1 - C\exp(-\Delta^2)$, it holds that
\begin{align*}
\begin{cases}
\vert D_{\leq t} \vert \leq \vert D_{\leq t-1} \vert + 20p( \vert D_{\leq t-1} \vert + \Delta^3)/(1+\delta) \qquad  \forall 1 \leq t \leq C\\
\vert D_0 \vert = 0
\end{cases}.
\end{align*}
Solving above recurrence, we have 
\begin{align*}
\Pr\left[ \forall 1 \leq t \leq C: \,\, \vert D_{\leq t} \vert \leq \Delta^3\left(\frac{20p}{1+\delta}+1\right)^t - \Delta^3 \right] \geq 1 - C\exp(-\Delta^2).
\end{align*}
If we take $\Delta > \left(\left(\frac{20p}{1+\delta}+1\right)^C - 1\right)^4$, then it holds that $\Delta^{13/4} > \Delta^3\left(\frac{20p}{1+\delta}+1\right)^C - \Delta^3$. Thus
\begin{align}
\label{eq-bad-2}
\Pr[ \bad_2 ] = \Pr\left[\vert D_{\leq C} \vert\geq \Delta^{13/4}\right] \leq \Pr\left[ \vert D_{\leq C}\vert \geq \Delta^3\left(\frac{20p}{1+\delta}+1\right)^C - \Delta^3\right] \leq C\exp(-\Delta^2).	
\end{align}
Finally, we define four more bad events $\badc_1, \badc_2, \badc_3, \badc_4$ as follows
\begin{itemize}
\item $\badc_4: \exists u \in V: \vert D_{\leq C} \cap B_{4}(u) \vert \geq \Delta^{13/4} = \Delta^{3 + 1/4}$.
\item For $k\in\{2,3\}$, define 
$\badc_{k} = \left(\bigcap_{ k<j \leq 4} \overline{\badc_{j}}\right) \cap \{\exists u \in V: \vert D_{\leq C} \cap B_k(u) \vert \geq \Delta^{k - 1 + 1 / k} \}.$
\item $\badc_1: \left(\bigcap_{1<j \leq 4}\overline{\badc_{j}}\right) \cap \{\exists u \in V: \vert D_{\leq C} \cap \Gamma(u) \vert \geq \zeta\Delta \}$.
\end{itemize}
If none of bad events $\badc_1, \badc_2, \badc_3,\badc_4$ occurs, then for all $u \in V$: $\vert D_{\leq C} \cap \Gamma(u) \vert < \zeta\Delta$. Thus we prove the Lemma by bounding the probability of bad events  $\badc_1, \badc_2, \badc_3,\badc_4$.

Note that the bad event $\badc_4$ implies the bad event $\bad_2$, thus by~\eqref{eq-bad-2} we have
\begin{align}
\label{eq-badc-4}
\Pr[\badc_4] \leq \Pr[\bad_2] \leq C\exp(-\Delta^2).
\end{align}
For $k =1, 2, 3$ we show that the bad event $\badc_k$ occurs with low probability. Assuming that none of events $\badc_j$ with $j > k$ occurs, otherwise the bad event $\badc_k$ can not occur. Fix a vertex $u \in V$,  let random variable $Z = \left\vert D_{\leq C} \cap B_k(u) \right\vert$ count the number of disagreements formed in $B_{k}(u)$ during time interval $[0,C]$. Let random variable $Z_t = \vert \left((X_t \oplus X^*_t) \setminus (X_{t-1} \oplus X^*_{t-1}) \right) \cap B_k(u)\vert$ count the number of new disagreements in $B_k(u)$ generated at step $t$. Since $X_0 = X^*_0$, then $Z \leq \sum_{t = 1}^C Z_t$.	
By~\eqref{eq-bad-1}, any disagreements can not percolate outside the ball $B_{R-3}(v)$ which implies $Z = 0$ if $u \not \in B_{R}(v)$. Assuming $u \in B_R(v)$, let us bound the expected value of each $Z_t$. As is stated in previous proof, a vertex $w \in B_k(u)$ satisfies $X_{t-1}(w) = X^*_{t-1}(w)$ but becomes a disagreement at step $t$ only if vertex $w$ is incident to bad edges, vertex $w$ is active at step $t$ and vertex $w$ proposes specific colors (which are determined by bad edges). We bound the total number of two types of bad edges incident to vertices in $B_k(u)$ at step $t$ as follows:
\begin{itemize}
\item {Type-1 bad edges within $B_k(u)$:} There are at most $\Delta^k$ edges with both endpoints in $B_k(u)$. Each of these edges should be counted as a bad edge at most once, because we only consider the type of bad edges that join an existing disagreement to a vertex $w \in B_k(u)$ such that $X^*_{t-1}(w) = X_{t-1}(w)$. 
\item {Type-1 bad edges at the boundary of $B_k(u)$:} Since none of bad events $\badc_j$ with $j > k$ occurs, then there are at most $\Delta^{k + 1/(k+1)}$ disagreements in $B_{k+1}(u) \setminus B_k(u)$. Each of disagreements has at most one neighbor in $B_k(u)$ because the girth is at least 9. There are at most $\Delta^{k + 1 + 1/k}$ such bad edges in total.
\item {Type-2 bad edges:} For each vertex $w \in B_k(u)$ such that $X^*_{t-1}(w) = X_{t-1}(w)$ , the event that  $w$ becomes a new disagreement at step $t$ may be caused by a directed edge $( w, w' )$ in graph $G^*$. By the definition of graph $G^*$, there is at most one such edge incident to each vertex $w$. Hence, the total number of such edges is at most $\Delta^k$.
\end{itemize}
Thus, the expected value of random variable $Z_t$ is upper bounded by
\begin{align*}
\E{Z_t} \leq \frac{2p(2\Delta^k + \Delta^{k + 1 / (k + 1)})}{q} \leq 	\frac{2p\left(2\Delta^{k-1} + \Delta^{k -1 + 1 / (k + 1)}\right)}{1+\delta}.
\end{align*}
Further, the laziness and proposed colors are fully independent, which implies $Z_t$ is stochastically dominated by the sum of independent random 0-1 variables.  For $k \in \{2, 3\}$, if we take large $\Delta$ such that $\Delta^{1/k} \geq \frac{20Cp(2 + \Delta^{1/(k+1)})}{1+\delta}$, then $\Delta^{k-1+1/k}/C \geq 10\E{Z_t}$. Thus by Chernoff bound~\eqref{eq-Chernoff-bound-2}, we have
\begin{align*}
\Pr\left[Z \geq \Delta^{k - 1 + 1 / k}\right]  \leq \Pr\left[\exists t: Z_t \geq \frac{\Delta^{k - 1 + 1 / k}}{C} \right]	\leq C\exp\left(-\frac{\Delta^{k - 1 + 1 / k}}{C} \right) \leq \exp(-\Delta\log \Delta).
\end{align*}
The last equality holds when $\Delta$ is sufficiently large such that $\Delta^{3/2} \geq C\ln C + C\Delta\log\Delta$. For $k=1$, we use Chernoff bound~\eqref{eq-Chernoff-bound-1}, then 
\begin{align*}
\Pr\left[Z \geq \zeta\Delta\right]\leq \Pr\left[\exists t: Z_t \geq \frac{\zeta\Delta}{C} \right] \leq C\left( \frac{2C\e p (2 + \Delta^{1/2})}{\zeta(1+\delta)\Delta} \right)^{\zeta\Delta/C} = \exp\left(- \Omega( \Delta \log \Delta ) \right),
\end{align*}
where the constant factor in $\Omega(\cdot)$ depends only on $C, p ,\delta, \zeta$.
For $k= 1,2,3$, take a union bound over the $\Delta^R = \Delta^{C+6}$ vertices $u \in B_R(v)$, then
\begin{align}
\label{eq-badc-k}
k=1, 2,3: \qquad\Pr[\badc_k] = \Delta^{C+6}\exp(-\Omega\left( \Delta\log \Delta\right)) = \exp(-\Omega\left( \Delta\log \Delta\right)),
\end{align}
where the constant factor in nation $\Omega(\cdot)$ depends only on $C, p ,\delta, \zeta$.
Summing the bounds in inequalities~\eqref{eq-badc-4} and \eqref{eq-badc-k} completes the proof.
\end{proof}

\subsection{Proof of local uniformity (Theorem~\ref{theorem-local-uniformity-G})}
Finally, the local uniformity property for the \LLM{} chain on graph $G$ can be proved by combining Lemma~\ref{lemma-local-uniformity-G*} and Lemma~\ref{lemma-compare-G-vs-G*}.

Let $\ell = \frac{1}{p}\left(\frac{1 + \delta}{\delta}\right)^2\ln\frac{1}{\zeta}$. Remark that $\ell$ is determined by $p,\delta,\zeta$. Consider any time $T \in [t_0, t_\infty]$, where $t_0 = \ell$ and $t_\infty = \exp(\Delta /C)$. Fix any coloring $X_{T - \ell}$ at time $(T-\ell)$, we apply the identical coupling $(X_t, X^*_t)$ for $ T -\ell \leq t \leq T$ from the same initial coloring $X_{T - \ell} = X^*_{T - \ell}$, where $X^*$ is given by the \LLM{} chain on graph $G^* = G_{\mathsf{in}}(v, 4)$. Note that, during the coupling, the \LLM{} on graph $G^*$ starts from the coloring $X_{T- \ell}$ and runs for $\ell$ steps. By Lemma~\ref{lemma-local-uniformity-G*}, if $\Delta \geq \Delta_1(\ell, p, \delta, \zeta) = \Delta_1(p,\zeta,\delta)$, then with probability at least $1 - \exp\left( -\zeta^2(1+\delta)\Delta/2 \right)$, we have
\begin{align}
\label{eq-lemma-local-uniformity-G-1}
\vert A(X_T^*, v) \vert > (q -\ell)\left(\frac{1-\gamma}{\mathrm{e}}\right)^{\deg(v)/(q - \ell)} - \zeta q,
\end{align}
where $\gamma = \exp{\left( - p \left(\frac{\delta}{1+\delta}\right)^2\ell\right)}+\frac{1}{q}\left(\frac{1+\delta}{\delta}\right)^2 $. By the definition of $\ell$, if we take $\Delta \geq \frac{1+\delta}{\zeta \delta^2}$, then $\zeta \leq \gamma \leq 2\zeta$. Note that $\gamma < 1$ because $\zeta < 1/10$. 
Furthermore, it holds that
\begin{align}
\label{eq-thm-8-1}
\Delta > \frac{2\ell}{(\delta+1)\zeta} \implies \frac{q - \ell}{q} \geq 1 - \zeta/2.
\end{align}
It can be verified that there exists $\Delta' = \Delta'(\delta, \zeta, p)$ such that if $\Delta \geq \Delta'$, it holds that
\begin{align}
\label{eq-thm-8-2}
\left( \frac{1 - \gamma}{\mathrm{e}} \right)	^{\frac{deg(v)}{q-\ell} - \frac{deg(v)}{q} } \geq \left( \frac{1 - 2\zeta}{\mathrm{e}} \right)	^{\frac{\ell}{(1+\delta)(q-\ell)}} \geq 1-\zeta/2.
\end{align}
Combining~\eqref{eq-thm-8-1},~\eqref{eq-thm-8-2} together, if $\Delta \geq \max\left\{\frac{1+\delta}{\zeta \delta^2}, \Delta' \right\}$, then it holds that
\begin{align}
\label{eq-lemma-local-uniformity-G-2}
(q -\ell)\left(\frac{1-\gamma}{\mathrm{e}}\right)^{\deg(v)/(q - \ell)} \geq (1 - \zeta/2)^2q\left(\frac{1-\gamma}{\mathrm{e}}\right)^{\deg(v)/q}.
\end{align}
Combining~\eqref{eq-lemma-local-uniformity-G-1} and~\eqref{eq-lemma-local-uniformity-G-2} implies
\begin{align}
\label{eq-thm-8-3}
\frac{\vert A(X_T^*, v) \vert}{q} &> (1 - \zeta/2)^2\left(\frac{1-\gamma}{\mathrm{e}}\right)^{\deg(v)/q} - \zeta \notag\\
(\deg(v) < q)\qquad&\geq (1 - \zeta)(1 - \gamma)\e^{-\deg(v)/q} - \zeta \notag\\
(\gamma < 2\zeta) \qquad &\geq (1 - 3\zeta)\e^{-\deg(v)/q} - \zeta.
\end{align}
Two chains are coupled for $\ell$ steps. Note that $\ell$ is determined by $p,\delta,\zeta$. By Lemma~\ref{lemma-compare-G-vs-G*}, if $\Delta \geq \Delta_2(\ell, p, \delta,\zeta) = \Delta_2(p,\delta,\zeta)$, then with probability at least $1 - \exp(-\Delta)$, it holds that
\begin{align*}
\vert A(X_T , v) \vert 	\geq  \vert A(X^*_T , v) \vert -  \zeta\Delta.
\end{align*}
Thus, condition on any coloring $X_{T - \ell}$, with probability at least $1 - \exp\left( -\zeta^2(1+\delta)\Delta/2 \right) - \exp(-\Delta)$, it holds that
\begin{align*}
\frac{\vert A(X_T, v) \vert}{q} &\geq 	\frac{\vert A(X_T^*, v) \vert}{q} - \frac{\zeta\Delta}{q}\\
(\text{By}~\eqref{eq-thm-8-3} \text{ and } q > \Delta)\qquad &\geq (1 - 3\zeta)\e^{-\deg(v)/q} - 2\zeta\\
(q > \deg(v)) \qquad &\geq (1 - 10\zeta)\e^{-\deg(v)/q}.
\end{align*}
By the law of total probability, summing over all possible coloring $X_{T-\ell}$ implies
\begin{align*}
\Pr\left[ \frac{\vert A(X_T, v) \vert}{q} \leq (1 - 10\zeta)\e^{-\deg(v)/q}  \right] &\leq  \exp\left( -\zeta^2(1+\delta)\Delta/2 \right) + \exp(-\Delta)\\
\left(C' = 2 / \min\{\zeta^2(1+\delta)/2, 1\}\right)\qquad &\leq 2\exp(-2\Delta/C')\\
(\Delta \geq C'\ln 2)\qquad&\leq \exp(-\Delta/C').
\end{align*}
Finally, let $C = 2C'$. The theorem is proved by taking a union bound over all the steps $T \in[t_0, t_{\infty}]$, where $t_0 = \ell = \frac{1}{p}\left(\frac{1 + \delta}{\delta}\right)^2\ln\frac{1}{\zeta}$ and $t_\infty = \exp(\Delta/C)$.

\section{Coupling with Local Uniformity}
\label{sec:Coupling-with-Local-Uniformity}
In this section, we use local uniformity property to avoid the worst case analysis in~\eqref{eq-worst-case-coupling-v0} and obtain a better mixing condition for graphs with girth at least 9 and maximum degree is greater than a sufficiently large constant. 

We define the constant $\alpha^* \approx 1.763$ to be the positive solution of 
\begin{align*}
\alpha^* = \e^{1/\alpha^*}.
\end{align*}
We consider $q$-colorings of graphs $G$ with maximum degree $\Delta$, where $q\ge(\alpha^*+\delta)\Delta$ for an arbitrary constant $\delta>0$. Without loss of generality we assume $\delta<0.3$ because bigger $\delta$ is already covered by Theorem~\ref{theorem-LLM-worst-case-coupling} on general graphs.
\begin{theorem}
\label{theorem-mixing-1.763}
For all $0<\delta<0.3$, there exists $\Delta_3=\Delta_3(\delta)$, $C'=C'(\delta)$, such that for every graph $G$ on $n$ vertices with maximum degree $\Delta\ge \Delta_3$ and girth $\ge 9$, if $q\geq (\alpha^* + \delta)\Delta$, then
the mixing rate of the \LLM{}  chain  with activeness $p = \frac{\delta}{30}$ on $q$-colorings of graph $G$   satisfies 
\[
\tau(\epsilon) \le C'\log \frac{n}{\epsilon}.
\]
\end{theorem}
Given the local uniformity property guaranteed by Theorem~\ref{theorem-local-uniformity-G}, the mixing rate in above theorem is proved by  following a similar framework as in~\cite{dyer2013randomly}.
We modify the framework to make it adaptive to the parallel chain, where the experiments carried on in a time scale of $O(n)$ steps in a sequential chain, are now in $O(1)$ steps, and a disagreement may percolate to many vertices in one step.

We begin with constructing a grand coupling of the \LLM{} as below.

\subsection{Coupling of arbitrary pair of colorings}
\label{subsec:grandcoupling}
In Section~\ref{section-coupling-on-general-graphs}, we give a local coupling $(X, Y) \rightarrow (X', Y')$ for $X, Y$ that differ only at a single vertex. Here, we use the path coupling to extend this coupling to a coupling of arbitrary pair of colorings.

Let $X, Y\in[q]^V$ be an arbitrary pair of colorings, not necessarily proper. 
Suppose that $X$ and $Y$ differ on precisely $\ell$ vertices $v_1, v_2 , \ldots , v_\ell$. A sequence of colorings $X = Z_0 \to Z_1 \to \ldots \to Z_\ell = Y$ is constructed as follows: for every $0\le i\le \ell$,
\begin{align*}
Z_i(v) = \begin{cases}
X(v)=Y(v) &\text{ if } v \not\in (X \oplus Y), \\
X(v) &\text{ if } v \in (X \oplus Y) \land v \in \{v_j \mid i < j \leq \ell\},\\
Y(v) &\text{ if } v \in (X \oplus Y) \land v \in \{v_j \mid 1 \leq j \leq i\}.
 \end{cases}
\end{align*}
Each coloring $Z_i$ is not necessarily proper, and $Z_{i - 1} \oplus Z_i = \{v_i\}$. A coupling $(X, Y) \rightarrow (X', Y')$ is then constructed by the path coupling:
\begin{itemize}
\item Sample a pair $(X', Z'_1)$ of colorings according to the local coupling $(X, Z_1) \rightarrow (X', Z_1')$ defined in Section~\ref{section-coupling-on-general-graphs}.
\item For $i=2, 3,\ldots, \ell$, conditioning on the sampled coloring $Z'_{i-1}$, sample coloring $Z_{i}'$ according to the local coupling $(Z_{i-1}, Z_i) \rightarrow (Z'_{i-1}, Z_{i}')$ defined in Section~\ref{section-coupling-on-general-graphs}. 
Finally, let $Y' = Z_\ell '$.
\end{itemize}
Note that the local coupling in Section~\ref{section-coupling-on-general-graphs} is constructed by coupling active vertex set and random proposed colors. 
Hence, a sequence of active vertex sets and random proposed colors $(\activeset_{Z_i}, \boldsymbol{c}_{Z_i})$  for $0 \leq i\leq\ell$ is constructed by the above process.

\begin{observation}
\label{observation-coupling}
Let $X, Y\in[q]^V$ be two colorings, and $X \oplus Y$ the set of vertices on which $X$ and $Y$ disagree. 
The followings hold for the coupling $(X, Y) \rightarrow (X', Y')$ defined above:
\begin{enumerate}
\item If $\dist_G(v, X \oplus Y) \geq 2$,
then $X'(v) = Y'(v)$.
\item If $\dist_G(v, X \oplus Y) = 1$, then $X'(v) \neq Y'(v)$ occurs only if vertex $v$ is active and proposes color $c_X(v) \in \{X(u), Y(u) \mid u \in \Gamma(v) \cap (X \oplus Y) \}$ in chain $X$.
\end{enumerate}
\end{observation}

\begin{proof}By the local coupling defined in Section~\ref{section-coupling-on-general-graphs}, two chains $X$ and $Y$ must select the same set of active vertices. If vertex $v \notin X \oplus Y$ is lazy in chain $X$, then it must be lazy in chain $Y$,  which implies that $X'(v) = X(v) = Y(v) = Y'(v)$.

Assume that vertex $v$ is active in chain $X$. We prove a stronger Claim which implies both claims in the observation:
\begin{claim*}
If $v \not \in (X \oplus Y)$ and $c_X(v) \not \in \{X(u), Y(u) \mid u \in\Gamma(v) \cap (X \oplus Y) \}$, then it must hold that $X(v) = Y(v)$ and $c_X(v) = c_Y(v)$.
\end{claim*}
Note that if $\dist_G(v, X \oplus Y) \geq 2$, then $c_X(v) \not \in \emptyset$ holds trivially. Thus it covers the first part of the Observation. Suppose that $\vert X \oplus Y \vert = k$, we prove it by induction on $k$. 
\begin{itemize}
\item Base case $k = 1$. Suppose $X \oplus Y = \{v_1\}$. If $\dist_G(v, X\oplus Y) \geq 2$, then by~\eqref{eq-worst-case-coupling-w}  in Lemma~\ref{lemma-worst-case-coupling-x'v-neq-y'v}, we have $X'(v) = Y(v)$, and by construction of the coupling, it holds that $c_X(v) = c_Y(v)$. If $\dist_G(v, X\oplus Y) = 1$, then by Observation~\ref{observation-worst-case-coupling} ,  $X'(v) \neq Y'(v) $ occurs only if $\{c_X(v), c_Y(v)\} \subseteq \{X({v_1}),Y({v_1})\}$. If $c_X(v) \not \in \{X({v_1}), Y({v_1})\}$, then  $X'(v) = Y'(v)$ and $c_X(v) = c_Y(v)$ regardless of which distribution $(c_X(v), c_Y(v))$ is sampled from. 
\item Suppose the Claim holds for all pairs of colorings that differ on no more than $k$ vertices. For any $X, Y$ differ at $k+1$ vertices, consider the coloring sequence $X=Z_0 \sim \ldots \sim Z_{k+1} = Y$. Then $(X', Z_k')$ is sampled from the coupling $(X, Z_k) \rightarrow (X', Z'_k)$ and $Y'$ is sampled from the coupling $(Z_k, Y) \rightarrow (Z'_k, Y')$ condition on $Z_k'$. Note that $v \notin (X \oplus Y)$ implies $v \notin (X \oplus Z_k)$ and $v \notin (Z_k \oplus Y)$.
Since $v \not \in (X \oplus Z_k)$ and $c_X(v) \not \in \{X(u), Z_k(u) \mid u \in\Gamma(v) \cap (X \oplus Z_k)\}$, then by I.H., it must hold that $X'(v) = Z_k'(v)$ and $c_X(v) = c_{Z_k}(v)$.  Further, since $v \not \in (Z_k \oplus Y)$ and $c_{Z_k}(v) =  c_X(v) \not \in \{Z_k(u), Y(u) \mid u \in\Gamma(v) \cap (Z_k \oplus Y)\}$, then by I.H., it must hold that $Z_k'(v) = Y'(v)$ and $c_{Z_k}(v) = c_Y(v)$. Combine them together, we have $X'(v) = Y'(v)$ and $c_X(v) = c_Y(v)$.  
\end{itemize}
\end{proof}

We then consider the coupling  $(X_t,Y_t)=(X_t,Y_t)_{t\ge 0}$ of two \LLM{} chains starting from initial colorings $(X_0,Y_0)$, constructed by applying the coupled transition $(X,Y)\to(X',Y')$ defined above at each step.

The following corollary of Observation~\ref{observation-coupling} says that if the initial colorings $X_0, Y_0$ differ at single vertex $v$, then disagreements can not percolate too fast in the coupled chains.

\begin{corollary}
\label{corollary-coupling-disagreement-percolation}
Let $v\in V$ and $X_0, Y_0\in[q]^V$ be two colorings that differ only at vertex $v$. For the coupling $(X_t, Y_t)$ of two \LLM{} chains, it holds that $X_t \oplus Y_t \subseteq B_t(v )$ for all $t\ge 0$. 	
\end{corollary}

The next corollary bounds the expectation and the deviation from expectation, for the number of new disagreements generated at each step of the couple chains. 

\begin{corollary}
\label{corollary-coupling-bound-new-disagreement}
Let $(X_t, Y_t)$ be the coupling of two \LLM{} chains with activeness $p$ on $q$-colorings of graph $G$ with maximum degree $\Delta$. Let $\mathcal{N}(D_{t}) = \vert(X_{t} \oplus Y_{t}) \setminus (X_{t-1} \oplus Y_{t-1})\vert$ be the number of new disagreements generated at step $t$. Then it holds that
\begin{align*}
\E{\mathcal{N}(D_{t}) \mid X_{t-1}, Y_{t-1}} \leq \frac{2p\Delta\vert X_{t-1} \oplus Y_{t-1} \vert}{q}.
\end{align*}
Furthermore, for any $\ell \geq \frac{20p\Delta\vert X_{t-1} \oplus Y_{t-1} \vert}{q}$, we have
\begin{align*}
\Pr\left[ \mathcal{N}(D_{t}) \geq \ell\mid X_{t-1}, Y_{t-1}\right] \leq \exp\left( -\ell \right).
\end{align*}
\end{corollary}
\begin{proof}
Let $\partial(X_{t-1} \oplus Y_{t-1}) = \{v \in V \mid v \not \in (X_{t-1} \oplus Y_{t-1}) \land \Gamma(v) \cap (X_{t-1} \oplus Y_{t-1}) \neq \emptyset\}$. From Observation~\ref{observation-coupling},  vertex $v$ becomes a new disagreement at step $t $ only if  $v \in \partial (X_{t-1} \oplus Y_{t-1})$, $v$ is active and proposes the color $c_X(v) \in \{X_{t-1}(u), Y_{t-1}(u) \mid u \in \Gamma(v) \cap (X_{t-1} \oplus Y_{t-1}) \}$. Hence, the expected number of new disagreements is at most $2p\Delta \vert X_{t-1} \oplus Y_{t-1} \vert/q$.

 Furthermore, the laziness and proposed colors are fully independently. Thus, the number of new disagreements is stochastically dominated by the sum of independent 0-1 random variables. The second inequality holds by the Chernoff bound~\eqref{eq-Chernoff-bound-2}.
\end{proof}

\subsection{Analysis of the coupling}
Next, we show that starting from any two colorings that differ at a single vertex $v$, after constant many steps of coupling, the Hamming distance contracts with a constant factor in expectation.
 
\begin{lemma}
\label{lemma-coupling-1-1/3}
For all $0 <\delta <0.3$, there exists $\Delta_4 = \Delta_4(\delta)$,  such that for every graph $G=(V,E)$ with maximum degree $\Delta \geq \Delta_4$ and girth at least $9$, if $q \geq (\alpha^*+\delta)\Delta$, then for any vertex $v\in V$, any initial colorings $X_0,Y_0\in[q]^V$ that differ only at $v$, the coupling $(X_t,Y_t)$ of two \LLM{} chains with activeness $p=\frac{\delta}{30}$ on $q$-colorings of graph $G$ satisfies
\begin{align*}
\E{X_{T_m} \oplus Y_{T_m}} \leq 1/3,
\end{align*}
where  $T_m = \frac{1200}{\delta^2}\ln \frac{600}{\delta}  = \Theta(1)$.
\end{lemma}

The main theorem regarding the mixing rate (Theorem~\ref{theorem-mixing-1.763}) is then an easy consequence of this lemma. 
Let $\Delta_3(\delta) = \Delta_4(\delta)$, where $\Delta_4$ is the threshold in Lemma~\ref{lemma-coupling-1-1/3}.  
For arbitrary two colorings  differ at a single vertex, there exists a coupling such that  the expected Hamming distance between them is at most $1/3$ after $C'' = \frac{1200}{\delta^2}\ln\frac{600}{\delta} = \Theta(1)$ steps. 
Since the Hamming distance between any two colorings is at most $n$, by the path coupling lemma, we have 
\[
\tau(\epsilon) \leq \frac{3C''}{2}\log \frac{n}{\epsilon}.
\]
The rest of the section is dedicated to the proof of Lemma~\ref{lemma-coupling-1-1/3}.

We use the technique developed in \cite{dyer2013randomly} to prove Lemma~\ref{lemma-coupling-1-1/3},
which partitions the time interval $[0,T_m]$ into to two disjoint phases $[0, T_b]$ and $[T_b + 1, T_m]$, the first phase is called the \emph{burn-in} phase.  After the burn-in phase, typically, the Hamming distance between two chains are bounded, all disagreements are near vertex $v$, and the local uniformity properties is guaranteed. Then we can prove that the expected Hamming distance will decrease in each step during $[T_b + 1, T_m]$. A crude upper bound is applied on the Hamming distance if non-typical events occur.
\begin{proof}[Proof of Lemma~\ref{lemma-coupling-1-1/3}]
For two colorings $X_t, Y_t$, define their difference as
\begin{align*}
D_t = \{u \mid X_t(u) \neq Y_t(u)\}	.
\end{align*}
Let $H_t = \vert D_t \vert$ denote their Hamming distance. Also, denote their cumulative difference by
\begin{align*}
D_{\leq t} = \bigcup_{t' \leq t}D_t,
\end{align*}
and denote their cumulative Hamming distance as $H_{\leq t} = \vert D_{\leq t} \vert$.

Let $\delta', p', \zeta'$ and $C'=C'(\delta', \zeta')$  denote the parameters $\delta, p, \zeta$ and $C=C(\delta,\zeta)$ in Theorem~\ref{theorem-local-uniformity-G}, respectively. 
We apply Theorem~\ref{theorem-local-uniformity-G} with $p' = p=\frac{\delta}{30}$, $\delta' = 1.7$ and $\zeta' = p /20$. 
Define
\begin{align*}
T_b = \frac{1}{p}\left(\frac{2.7}{1.7}\right)^2\ln\frac{20}{p}.
\end{align*}
Note that $C'=C'(\delta', \zeta')$ now depends only on $\delta$. Recall that $T_m = \frac{1200}{\delta^2}\ln \frac{600}{\delta}$. If  we take $\Delta \geq C'\ln T_m$, then $T_m  < \exp(\Delta / C')$. Thus we can assume that the local uniformity property in Theorem~\ref{theorem-local-uniformity-G} holds for all  time $t \in [T_b, T_m]$. If $\Delta \geq \Delta_0(p, \delta',\zeta') = \Delta_0(\delta)$, then it holds that
\begin{align}
\label{eq-coupling-local-uniformity}
\Pr\left[ \forall t \in \left[  t_b , t_m \right]: \,\, \frac{\vert A(X_t, v) \vert }{q} \ge (1 - p/2)\e^{-\deg(v) / q} \right] \ge 1- \exp(-\Delta / C').
\end{align}

For each $t \geq T_b$, we define following bad events:
\begin{itemize}
\item $\mathcal{E}(t)$: there exists some time $s < t$, such that $\vert X_s \oplus Y_s \vert > \Delta^{2/3}$.
\item $\mathcal{B}_1 (t)$: $D_{\leq t} \not \subseteq B_{T_m}(v)$.
\item $\mathcal{B}_2(t)$: there exists some time $T_b \leq \tau \leq t$ and a vertex $z \in B_{T_m}(v)$ such that 
\begin{align*}
\vert A(X_\tau, z) \vert \leq (1 - p / 2)q\e^{-d(z) / q}.	
\end{align*}
\end{itemize}
Define bad event $\bad(t)$ as
\begin{align*}
\bad(t) = \bad_1(t) \cup \bad_2(t).	
\end{align*}
Define good event $\mathcal{G}(t)$ as
\begin{align*}
\mathcal{G}(t) = \overline{\mathcal{E}(t)} \cap \overline{\bad(t)}.	
\end{align*}
For all events when the time $t$ is dropped, we are referring to the event at time $t = T_m$. Then the Hamming distance between $X_{T_m}$ and $Y_{T_m}$ can be bounded as follows
\begin{align}
\label{eq-bound-HTM-by-cases}
\E{H_{T_m}} &= \E{H_{T_m}\one{\mathcal{E}}} + \E{H_{T_m}\one{\overline{\mathcal{E}}}\one{\mathcal{B}}} + \E{H_{T_m}\one{\mathcal{G}}}\notag\\
&\leq \E{H_{T_m}\one{\mathcal{E}}} + \Delta^{2/3}\Pr[\bad] + \E{H_{T_m}\one{\mathcal{G}}}.
\end{align}
Since the bad events (non-typical events) occur with small probability, then we have following Claims.
\begin{claim}
\label{claim-bound-HTM-by-cases-12}
$\Pr[\bad] \leq \exp(-\sqrt{\Delta})$ and $\E{H_{T_m}\one{\mathcal{E}}} \leq \exp(-\sqrt{\Delta})$.	
\end{claim}
If the good event (typical event) $\mathcal{G}$ occurs, then we can use local uniformity property to prove that the Hamming distance decreases by a constant factor during $[T_b + 1, T_m]$. Thus we have following Claim.
\begin{claim}\label{claim-bound-HTM-by-cases-3}
$\E{H_{T_m}\one{\mathcal{G}}} \leq 1 / 9$.
\end{claim}
Lemma~\ref{lemma-coupling-1-1/3} follows by combining~\eqref{eq-bound-HTM-by-cases}, Claim~\ref{claim-bound-HTM-by-cases-12} and Claim~\ref{claim-bound-HTM-by-cases-3}.
\end{proof}

\begin{proof}[Proof of Claim~\ref{claim-bound-HTM-by-cases-12}]
At first, we prove that
\begin{align*}
\Pr[\bad] \leq \exp(-\sqrt{\Delta}).	
\end{align*}
By Corollary~\ref{corollary-coupling-disagreement-percolation}, we know that disagreements can not percolate outside the ball $B_{T_m}(v)$, which implies $\Pr[\bad_1] = 0$. The probability of bad event $\mathcal{B}_2$ can be bounded by~\eqref{eq-coupling-local-uniformity}. Thus, we have
\begin{align*}
\Pr[\bad] &= \Pr[\bad_1] + \Pr[\bad_2]\\
&= \Pr[\bad_2]\\
(\ast)\qquad&\leq \Delta^{T_m}\exp(-\Delta/ C')\\
\text{}\qquad&\leq \exp(-\sqrt{\Delta}),
\end{align*}
where inequality ($\ast$) is a union bound over all vertices $z \in B_{T_m}(v)$. The last inequality holds for sufficiently large $\Delta$ such that $\Delta \geq C'(T_m \ln \Delta + \sqrt{\Delta})$. Note the $C'$ and $T_m$ depends only on $\delta$. 

Next, we prove that
\begin{align*}
\E{H_{T_m}\one{\mathcal{E}}} \leq \exp(-\sqrt{\Delta}).	
\end{align*}
We will prove that for every $\ell \geq \Delta^{2/3}$, there exists $C'' = C''(\delta)>0$ such that
\begin{align}
\label{eq-up-bound-too-many-disagree}
\Pr[H_{\leq T_m} \geq \ell] \leq \exp(-C''\ell).	
\end{align}
Then, we bound the expected Hamming distance between $X_{T_m}$ and $Y_{T_m}$ as follows
\begin{align*}
\E{H_{T_m}\one{\mathcal{E}}} &\leq \E{H_{\leq T_m} \one{\mathcal{E}}}\\
\text{(By definition of $\mathcal{E}$)}\qquad&\leq  \sum_{\ell \geq \Delta^{2/3}}\ell \Pr[H_{\leq T_m} =\ell]\\
&= \Delta^{2/3}\Pr[H_{\leq T_m} \geq \ell] + \sum_{\ell \geq \Delta^{2/3}+1}\Pr[H_{\leq T_m} \geq \ell]\\
&\leq \Delta^{2/3}\sum_{\ell \geq \Delta^{2/3}}\Pr[H_{\leq T_m} \geq \ell]\\
(\text{By~\eqref{eq-up-bound-too-many-disagree} })\qquad &\leq \Delta^{2/3}\sum_{\ell \geq \Delta^{2/3}}\exp(-\ell C'')\\
&= \frac{\Delta^{2/3}\exp(-\Delta^{2/3}C')}{1 - \exp(-C'')}\\
\qquad&\leq \exp(-\sqrt{\Delta}).
\end{align*}
The last inequality holds for large $\Delta$ such that $C''\Delta^{2/3} \geq \frac{2}{3}\ln \Delta + \sqrt{\Delta} - \ln (1-\exp(-C''))$.

Now we prove inequality~\eqref{eq-up-bound-too-many-disagree}. Define a sequence $c_0, c_1, \ldots, c_{T_m}$ as follows
\begin{itemize}
\item $c_{T_m} = 1$;
\item For each $1 \leq t \leq T_m$,  $c_t = \left( 1 + 12p \right)c_{t - 1} = \left( 1 + \frac{2\delta}{5}\right)c_{t - 1}.$
\end{itemize}
For every $\ell \geq \Delta^{2/3}$, we bound the probability of the event $H_{\leq t} \geq c_t\ell$ for $0 \leq t \leq T_m$, where $T_m = \frac{1200}{\delta^2}\ln\frac{600}{\delta}$. Note that $H_{\leq 0} = 1$, if we take $\Delta > \left(1 + \frac{2\delta}{5}\right)^{3T_m/2}$, then
\begin{align}
\label{eq-bound-h0}
\Pr\left[  H_{\leq 0} \geq c_0 \ell\right] = 0.	
\end{align}
Then for each $1 \leq t \leq T_m$, by the law of total probability, we have
\begin{align*}
\Pr\left[  H_{\leq t} \geq c_t \ell\right] =& \Pr[H_{\leq t} \geq c_t \ell \mid H_{\leq t-1} \geq c_{t-1}\ell]\Pr[H_{\leq t-1} \geq c_{t-1}\ell]\\
&+\Pr[H_{\leq t} \geq c_t \ell \mid H_{\leq t-1} < c_{t-1}\ell]\Pr[H_{\leq t-1} < c_{t-1}\ell]\\
\leq& \Pr[H_{\leq t-1} \geq c_{t-1}\ell] + \Pr[H_{\leq t} \geq c_t \ell \mid H_{\leq t-1} < c_{t-1}\ell]
\end{align*}
Let $\mathcal{N}(D_t) = \vert (X_t \oplus Y_t ) \setminus (X_{t-1} \oplus Y_{t-1}) \vert$ be the number of new disagreements generated at step $t$, then it holds that
\begin{align*}
 \Pr[H_{\leq t} \geq c_t \ell \mid H_{\leq t-1} < c_{t-1}\ell] &\leq \Pr[\mathcal{N}(D_t) \geq (c_t - c_{t-1})\ell \mid H_{\leq t-1} < c_{t-1}\ell ]\\
(c_t = \left( 1 + 12p \right)c_{t - 1}) \qquad&= \Pr[\mathcal{N}(D_t) \geq 12p c_{t-1}\ell \mid H_{\leq t-1} < c_{t-1}\ell ]\\
&\leq \exp(-12p c_{t-1} \ell). 
\end{align*}
The last inequality is due to Corollary~\ref{corollary-coupling-bound-new-disagreement} (Note that $q \geq 1.7\Delta$ and $\vert X_{t-1} \oplus Y_{t-1} \vert \leq H_{\leq t-1}$). Thus
\begin{align}
\label{eq-bound-ht}
\Pr\left[  H_{\leq t} \geq c_t \ell\right] \leq \Pr[H_{\leq t-1} \geq c_{t-1}\ell] + \exp(-12p c_{t-1} \ell).
\end{align}
Combining~\eqref{eq-bound-h0}, ~\eqref{eq-bound-ht} and the definition of sequence $c$ (note that $c_{T_m} = 1$) implies 
\begin{align*}
\Pr\left[  H_{\leq T_m} \geq \ell\right] \leq \sum_{i = 1}^{T_m}	\exp(-12p c_{t-1} \ell) \leq T_m\exp(-12pc_0 \ell) = \exp\left(-12pc_0\ell + \ln T_m \right).
\end{align*}
Note that $\ell \geq \Delta^{2/3}$ and $c_0 = (1+12p)^{-T_m}$. If $\Delta \geq \left( \frac{\ln T_m}{11pc_0} \right)^{3/2}$ (note that $T_m, c_0, p$ depend only on $\delta$), then we have $-12pc_0\ell + \ln T_m \leq -pc_0\ell$, which implies
\begin{align*}
\Pr\left[  H_{\leq {T_m}} \geq \ell\right] \leq \exp(-p c_{0} \ell) = \exp\left(- \frac{p\ell}{(1+12p)^{T_m}} \right) = \exp(-\ell C'').
\end{align*}
Recall that $p = \frac{\delta}{30}$ and $T_m= \frac{1200}{\delta^2}\ln \frac{600}{\delta}$, thus $C'' = C''(\delta)$.
This proves inequality~\eqref{eq-up-bound-too-many-disagree}.
\end{proof}

\begin{proof}[Proof of Claim~\ref{claim-bound-HTM-by-cases-3}]
Condition on $X_{t}, Y_{t}$, we will bound the expected value of $H_{t+1}$ by path coupling. Suppose $X_t, Y_t$ differ at $h$ vertices $v_1,v_2,\ldots,v_h$. Then, according to the coupling, we construct a sequence of colorings $X=Z_0\sim Z_1 \sim \ldots \sim Z_h = Y$, such that each $Z_i$ and $Z_{i-1}$ differ only at vertex $v_i$. Consider the coupling $(Z_{i-1}, Z_i) \rightarrow (Z'_{i-1}, Z'_i)$, by Lemma~\ref{lemma-worst-case-coupling-x'v-neq-y'v}, we have
\begin{align*}
\E{\vert Z'_{i-1} \oplus Z'_i \vert \mid Z_{i-1}, Z_i} &\leq 1 - \frac{p(q - \Delta)}{q}\left(1 - \frac{3p}{q}\right)^\Delta + \frac{p\Delta}{q}\\	
\left(q > \alpha^*\Delta\right)\qquad &\leq 1 + \frac{p}{\alpha^*}
\end{align*}
Therefore, give $X_t,Y_t$, the expected value of $H_{t+1}$ can be bounded by triangle inequality as follows
\begin{align}
\label{eq-burn-in}
\E{H_{t+1} \mid H_t} \leq \left(1+ \frac{p}{\alpha^*} \right)H_t.
\end{align}
The inequality shows that the number of disagreements increases in each step. However, this bound will only be used during the burn-in phase $[0, T_b]$.

For each time $t \in [T_b, T_m]$, given $X_t, Y_t$,  assuming the good event $\mathcal{G}(t)$ occurs, we bound the the expected value of $H_{t+1}$ by path coupling. Suppose $X_t, Y_t$ differ at $h$ vertices $v_1,v_2,\ldots,v_h$. According to the coupling, we construct the path $X=Z_0\sim Z_1 \sim \ldots \sim Z_h = Y$. Since we assume that the good event $\mathcal{G}(t)$ occurs, then for each $0\leq i \leq h$, it holds that $\vert X \oplus Z_i \vert \leq \Delta^{2/3}$, $v_i \in B_{T_m}(v)$ and $\vert A(X, v_i)\vert \geq (1 - p / 2)q\e^{-\deg(v_i) / q}$. Thus we have
\begin{align*}
\vert A(Z_i, v_i) \vert \geq \vert A(X, v_i) \vert - \Delta^{2/3} \geq 	 (1 - p / 2)q\e^{-\deg(v_i) / q} - \Delta^{2/3} \geq (1 - p / 2)q\e^{-\Delta / q} - \Delta^{2/3}.
\end{align*}
Together with inequalities~\eqref{eq-worst-case-coupling-u} and~\eqref{eq-worst-case-coupling-w}, we have 
\begin{align*}
\E{\vert Z'_{i-1} \oplus Z'_i \vert \mid Z_{i-1}, Z_i} &\leq 1 - \frac{p\vert A(Z_i, v_i) \vert}{q}\left(1 - \frac{3p}{q}\right)^\Delta + \frac{p\Delta}{q}\\	
&\leq 1 - p\left((1-p/2)\e^{-\Delta/q}- \frac{1}{\alpha^*\Delta^{1/3}}\right)\left(1 - \frac{3p}{\alpha^*}\right) + \frac{p\Delta}{q},
\end{align*}
where the last inequality is because $q > \alpha^*\Delta$ and $\left(1 - \frac{3p}{q}\right)^\Delta \geq 1 - \frac{3p}{\alpha^*}$ due to Bernoulli's inequality. 
Note that, if we take $\Delta \geq \left( \frac{2\mathrm{e}^{1/\alpha^*}}{p\alpha^*} \right)^3 \geq \left( \frac{2\mathrm{e}^{\Delta/q}}{p\alpha^*} \right)^3$, then  $\frac{1}{\alpha^*\Delta^{1/3}} \leq \frac{p}{2}\mathrm{e}^{-\Delta/q}$.
It holds that 
\begin{align*}
\E{\vert Z'_{i-1} \oplus Z'_i \vert \mid Z_{i-1}, Z_i} &\leq 1 - p(1-p)\e^{-\Delta/q}\left(1 - \frac{3p}{\alpha^*}\right) + \frac{p\Delta}{q}\\
&\leq 1 - p \left( (1-3p)\e^{-1 / (\alpha^*+\delta)} - \frac{1}{\alpha^* + \delta} \right)\\
&=1-p\left(\left( \mathrm{e}^{-1/(\alpha^* + \delta)} - \frac{1}{\alpha^*+\delta} \right)-3p\mathrm{e}^{-\frac{1}{\alpha^*+\delta}}\right)\\
&\le 1-p\left(\frac{\delta}{5}-3p\right),
\end{align*}
where the last inequality is because for $0 < \delta < 0.3$, $\mathrm{e}^{-1/(\alpha^* + \delta)} - \frac{1}{\alpha^*+\delta} \geq \frac{\delta}{5}$ and $\mathrm{e}^{-\frac{1}{\alpha^*+\delta}}\le 1$.

For $p=\frac{\delta}{30}$,
it holds that
\begin{align*}
\E{\vert Z'_{i-1} \oplus Z'_i \vert \mid Z_{i-1}, Z_i} \leq 1 - \frac{\delta^2}{300}	.
\end{align*}
Hence, for each $t \in [T_b, T_m-1]$, given $X_t, Y_t$, assuming the good event $\mathcal{G}(t)$ holds, we have
\begin{align}
\label{eq-contraction-in-good-case}
\E{H_{t+1} \mid X_t,Y_t } \leq \left( 1 - \frac{\delta^2}{300}	 \right) H_t.	
\end{align}
For each $t \in [T_b, T_m - 1]$, it holds that
\begin{align*}
\E{H_{t+1}\one{\mathcal{G}(t)}}	&= \E{\E{H_{t+1}\one{\mathcal{G}(t)}\mid X_0, Y_0, \ldots, X_t, Y_t}}\notag\\
(\ast)\qquad&\leq \E{\E{H_{t+1}\mid X_0, Y_0, \ldots, X_t, Y_t}\one{\mathcal{G}(t)}}\notag\\
\text{(By~\eqref{eq-contraction-in-good-case})}\qquad&\leq \left( 1 - \frac{\delta^2}{300}	 \right)\E{H_t \one{\mathcal{G}(t)}}\notag\\
(\ast\ast)\qquad&\leq  \left( 1 - \frac{\delta^2}{300}	 \right)\E{H_t \one{\mathcal{G}(t-1)}}.
\end{align*}
Inequality $(\ast)$ is because the event $\mathcal{G}(t)$ is determined by $X_0,Y_0,\ldots,X_t, Y_t$. Inequality $(\ast\ast)$ is because the event $\mathcal{G}(t)$ implies the event $\mathcal{G}(t-1)$. By induction, it holds that
\begin{align*}
\E{H_{T_m}\one{\mathcal{G}}} \leq \E{H_{T_m}\one{\mathcal{G}(T_m - 1)}}\leq \left( 1 - \frac{\delta^2}{300} \right)^{T_m - T_b} \E{H_{T_b}\one{\mathcal{G}(T_b - 1)}}.
\end{align*}
Note that $\E{H_{T_b}\one{\mathcal{G}(T_b - 1)}} \leq \E{H_{T_b}}$, and apply~\eqref{eq-burn-in} for $t \in [0, T_b - 1]$, we have
\begin{align}
\E{H_{T_m}\one{\mathcal{G}}} &\leq \left( 1 - \frac{\delta^2}{300}\right)^{T_m - T_b} \left(1+ \frac{p}{\alpha^*}\right)^{T_b}H_0.
\end{align}
Note that $1 < \delta < 0.3$, $p = \frac{\delta}{30}$. It holds that $T_b = \frac{1}{p}\left(\frac{2.7}{1.7}\right)^2\ln\frac{20}{p} \leq \frac{120}{\delta}\ln\frac{600}{\delta}$. Since $T_m = \frac{1200}{\delta^2}\ln\frac{600}{\delta}$, then we have $T_m - T_b \geq \frac{900}{\delta^2}\ln\frac{600}{\delta}$. Note that $H_0 = 1$. We have
\begin{align*}
\E{H_{T_m}\one{\mathcal{G}}}  &\leq 	\left( 1 - \frac{\delta^2}{300} \right)^{\frac{900}{\delta^2}\ln\frac{600}{\delta}} \left(  1+ \frac{p}{\alpha^*}\right )^{\frac{4}{p}\ln \frac{600}{\delta}}\\
&\leq \left( \frac{\delta}{600} \right)^{3 - 4/\alpha^*}\\
&\leq \frac{1}{9}.
\end{align*}
\end{proof}


\end{document}